\documentclass[a4paper,twocolumn,11pt,accepted=2022-03-21,superscriptaddress,amsmath,amssymb,aps,dvipsnames,table]{quantumarticle}
\pdfoutput=1
\usepackage{algorithmic}

\usepackage[numbers,sort&compress]{natbib}

\usepackage{dcolumn} 
\usepackage{hyperref}
\usepackage[pdftex]{graphicx}
\usepackage{float}
\usepackage{adjustbox}

\usepackage[table]{xcolor}
\definecolor{blue1}{RGB}{210, 240, 255}
\definecolor{blue2}{RGB}{160, 220, 255}
\definecolor{blue3}{RGB}{130, 200, 255}
\definecolor{red1}{RGB}{255, 215, 215}
\definecolor{red2}{RGB}{255, 200, 200}
\definecolor{red3}{RGB}{255, 175, 175}

\newcommand{\cbOne}{\cellcolor{blue1}}
\newcommand{\cbTwo}{\cellcolor{blue2}}
\newcommand{\cbThree}{\cellcolor{blue3}}
\newcommand{\crOne}{\cellcolor{red1}}
\newcommand{\crTwo}{\cellcolor{red2}}
\newcommand{\crThree}{\cellcolor{red3}}
\usepackage{extpfeil} 

\usepackage{mathtools,mathdots}
\usepackage{physics}
\usepackage{dsfont}
\usepackage{bm} 

\usepackage{array,multirow,makecell} 

\usepackage{colortbl}

\makeatletter

    \def\CT@@do@color{%
      \global\let\CT@do@color\relax
            \@tempdima\wd\z@
            \advance\@tempdima\@tempdimb
            \advance\@tempdima\@tempdimc
    \advance\@tempdimb\tabcolsep
    \advance\@tempdimc\tabcolsep
    \advance\@tempdima2\tabcolsep
            \kern-\@tempdimb
            \leaders\vrule
                    \hskip\@tempdima\@plus  1fill
            \kern-\@tempdimc
            \hskip-\wd\z@ \@plus -1fill }
    \makeatother

\usepackage{tikz}
\usetikzlibrary{positioning}
\usetikzlibrary{math} 
\usetikzlibrary{fit}

\usepackage{qcircuit}

\usepackage{amsthm}

\newtheorem{thm}{Theorem}[section]
\newtheorem{prop}{Proposition}[section]
\newtheorem{conj}{Conjecture}[section]

\theoremstyle{definition}

\usepackage{mleftright}
\mleftright

\def\beq{\begin{equation}}
\def\eeq{\end{equation}}
\def\beqa{\begin{eqnarray}}
\def\eeqa{\end{eqnarray}}

\newcommand{\bb}[1]{{\mathbb #1}}

\newcommand{\diag}[9]
{
\resizebox{3cm}{!}{
\framebox{
\begin{tikzpicture}
\tikzmath{\w1 = 0.8; \w2 =0.8; \sc=2;} 

\node[draw] at (2.5/\sc,2.5/\sc) {$ #1 $}; 
\node[] at (0.5/\sc,1.5/\sc) {$ #2 $}; 
\node[] at (1.5/\sc,1.5/\sc) {$ #3 $}; 
\node[] at (3.5/\sc,1.5/\sc) {$ #4 $}; 
\node[] at (4.5/\sc,1.5/\sc) {$ #5 $}; 
\node[] at (2.5/\sc,0.5/\sc) {$ #6 $}; 
\node[draw] at (1/\sc,0) {$ #7 $}; 
\node[draw] at (4/\sc,0) {$ #8 $}; 
\node[] at (2.5/\sc,-0.5/\sc) {$ #9 $}; 

\end{tikzpicture}
}
}
}

\begin{document}


\title{Solving correlation clustering with QAOA and a Rydberg qudit system: a full-stack approach}

\author{Jordi R. Weggemans}
\affiliation{CWI, Science Park 123, 1098 XG Amsterdam, The Netherlands}
\affiliation{QuSoft, Science Park 123, 1098 XG Amsterdam, The Netherlands}

\author{Alexander Urech}
\affiliation{Van der Waals-Zeeman Institute, Institute of Physics, University of Amsterdam, Science Park 904, 1098 XH Amsterdam, the Netherlands}
\affiliation{QuSoft, Science Park 123, 1098 XG Amsterdam, The Netherlands}

\author{Alexander Rausch}
\affiliation{Robert Bosch GmbH, Corporate Research, Robert-Bosch-Campus 1, 71272 Renningen, Germany}

\author{Robert Spreeuw}
\affiliation{Van der Waals-Zeeman Institute, Institute of Physics, University of Amsterdam, Science Park 904, 1098 XH Amsterdam, the Netherlands}
\affiliation{QuSoft, Science Park 123, 1098 XG Amsterdam, The Netherlands}

\author{Richard Boucherie}
\affiliation{Stochastic Operations Research, Department of Applied Mathematics,
University of Twente, 7500 AE, Enschede, The Netherlands.}

\author{Florian Schreck}
\affiliation{Van der Waals-Zeeman Institute, Institute of Physics, University of Amsterdam, Science Park 904, 1098 XH Amsterdam, the Netherlands}
\affiliation{QuSoft, Science Park 123, 1098 XG Amsterdam, The Netherlands}

\author{Kareljan Schoutens}
\affiliation{Institute for Theoretical Physics, University of Amsterdam, Science Park 904, 1098 XH Amsterdam, the Netherlands}
\affiliation{QuSoft, Science Park 123, 1098 XG Amsterdam, The Netherlands}

\author{Ji\v{r}\'{i} Min\'{a}\v{r}}
\affiliation{Institute for Theoretical Physics, University of Amsterdam, Science Park 904, 1098 XH Amsterdam, the Netherlands}
\affiliation{QuSoft, Science Park 123, 1098 XG Amsterdam, The Netherlands}

\author{Florian Speelman}
\affiliation{Informatics Institute, University of Amsterdam, Science Park 904, 1098 XH Amsterdam, the Netherlands}
\affiliation{QuSoft, Science Park 123, 1098 XG Amsterdam, The Netherlands}

\begin{abstract}


We study the correlation clustering problem using the quantum approximate optimization algorithm (QAOA) and qudits, which constitute a natural platform for such non-binary problems. Specifically, we consider a neutral atom quantum computer and propose a full stack approach for correlation clustering,
including Hamiltonian formulation of the algorithm, analysis of its performance, identification of a suitable level structure for ${}^{87}{\rm Sr}$ and specific gate design. 
We show the qudit implementation is superior to the qubit encoding as quantified by the gate count.
For single layer QAOA, we also prove (conjecture) a lower bound of $0.6367$ ($0.6699$) for the approximation ratio on 3-regular graphs. 
Our numerical studies evaluate the algorithm's performance by considering complete and Erd\H{o}s-R\'enyi graphs of up to 7 vertices and clusters.
We find that in all cases the QAOA surpasses the Swamy bound $0.7666$ for the approximation ratio for QAOA depths $p \geq 2$. 
Finally, by analysing the effect of errors when solving complete graphs we find that their inclusion severely limits the algorithm's performance.

%
\end{abstract}

\maketitle

\section{\label{sec:Int} Introduction}
The Quantum Approximate Optimization Algorithm (QAOA) is a promising attempt at trying to find a quantum advantage when using near-term Noisy Intermediate-Scale Quantum (NISQ) devices~\cite{Farhi2014}. The current body of literature points into mixed directions as far as the utility of QAOA is concerned: whilst it has provable advantages such as recovering near optimal query complexity in Grover's search~\cite{Jiang2017}, exhibiting universality~\cite{Lloyd2018,Morales2020} and the possibility for quantum supremacy~\cite{Farhi2016}, there are also known limitations in the low depth regime~\cite{Akshay2020,Hastings2019,Bravyi2020_2,Marwaha2021}. However, current analytical tool for analysing the performance of QAOA have only been able to investigate very specific problem instances, predominantly at low depth~\cite{Farhi2019,Wauters2020,Claes2021,Wurtz2020,wang2018quantum}. A general analytical approach remains to be found, which is why a large portion of the literature resorts to numerical simulation. 

Earlier work on QAOA studied problems such as MAXCUT~\cite{Farhi2014,wang2018quantum,guerreschi2019qaoa} and MAX E3LIN2~\cite{farhi2014quantum}, where the translation procedure entails translating constraints on a binary string into a cost Hamiltonian.  In the current work we focus on the \emph{correlation clustering} problem. First introduced by Bansal et al.~\cite{Bansal2004}, reductions to correlation clustering are common in machine-learning contexts, with applications ranging from data analysis~\cite{benjelloun2009swoosh} to image recognition and pose estimation~\cite{InsafPAAS16}.
Since it is a computationally intensive task to find these clusterings, with approximations being APX-hard~\cite{Charikar2005}, it is natural to ask whether quantum algorithms can do better.

For much of the earlier work, the QAOA blueprint entailed working with qubits: the +1/-1 eigenstates of the Pauli Z-operator are chosen in direct correspondence with the variables of the underlying combinatorial optimization problem to compute the cost function, and the mixer Hamiltonian is chosen as an independent Pauli X applied to each qubit. For instance, the MAXCUT problem divides a given graph into two sets, which makes the two-qubit ZZ-operator suitable for encoding the objective function.
Our situation is quite different as correlation clustering naturally encompasses many more different clusters.
While it is always possible in principle to encode higher-dimensional information into multiple qubits, this does significantly complicate the interactions required.

In this paper we investigate a different route, namely encoding the cluster choice into a qudit. Qudits can be realized in a number of physical platforms including photons~\cite{Bent_2015_PRX, kues2017chip, Babazadeh_2017_PRL, islam2017provably, Yoshikawa_2018_PRA, Lu_2018_PRL, imany2018,wang2018multidimensional,reimer2019high}, ions~\cite{Senko_2015_PRX, Randall_2015_PRA, Leupold_2018_PRL}, molecules~\cite{hussain2018coherent, sawant2020ultracold, Albert_2020_PRX}, superconducting circuits~\cite{neeley2009emulation, svetitsky2014hidden, Tan_2021_PRL}, nuclear magnetic resonance platforms~\cite{dogra2014determining, gedik2015computational} or NV centres~\cite{moro2019realization, soltamov2019excitation} and can offer a more resource efficient approach to quantum computing as compared to qubits~\cite{wang2020qudits}. Importantly, for the present problem of correlation clustering, they offer a support with the Hilbert space which is \emph{native} to the studied problem, namely the different qudit states can encode the labelling of nodes into their respective cluster. 

In this respect, quantum computers based on trapped neutral atoms interacting via highly excited Rydberg states are of particular interest~\cite{Saffman_2010_RMP,saffman2016quantum}. These platforms underwent striking developments in recent years with advances such as the creation of arbitrary quantum processor geometries~\cite{Barredo2016,Endres_2016_Science,Barredo_2018_Nature} containing hundreds of individual atoms~\cite{scholl2020programmable,semeghini2021probing} or the efficient implementation of quantum gates~\cite{Levine_2019_PRL,Madjarov_2020_NatPhys}. This in principle allows for efficient implementation of QAOA algorithms~\cite{Zhou2020,Dalyac2021} and led to intense efforts, both academic and industrial~\cite{Pasqal,AtomComputing,ColdQuanta,QuEra}, in neutral atom based quantum computing~\cite{Henriet2020,Morgado2020}. 

Even though QAOA was invented as a NISQ-suitable algorithm, it can be far from trivial to connect the abstract Hamiltonian formulation to a physical system.
Problems include how to relate the intended Hamiltonians to manipulations on the system, how to manage the spatial structure and non-uniformity of the pairwise interactions of a system in a lattice, and careful analysis of the types of errors that might occur in these operations.

It has also become progressively clear that an efficient quantum algorithm has to be designed in a way which implicitly takes into account the relevant hardware constraints, a so-called full-stack approach~\cite{Alexeev2019}. In this work we describe precisely such a full-stack solution starting from a combinatorial optimization problem (the correlation clustering problem), through a QAOA Hamiltonian formulation, all the way to describing how to control and analyse a Rydberg qudit system. This entails:
\begin{itemize}
\item We encode the correlation clustering problem into the QAOA paradigm, specifically tailored for qudit quantum systems. We give several improvements to the vanilla QAOA (see Sec.~\ref{sec:Recap} for the definition), guided by simulations, including experiments with various meta-optimization strategies.
\item Next, we go from the abstract Hamiltonian formulation to the operations available for an actual abstract Rydberg system.
This entails showing how to drive the Rydberg system in a way that corresponds to the cost and mixer Hamiltonians, where we also have to take care of the spatial aspect of the interactions.
Here, we focus on the example of fermionic strontium $^{87}{\rm Sr}$, but a similar derivation would apply to any related system.
\item We analyse the algorithm performance in the presence of noise corresponding to the 'random Pauli' method and link this error model to actual errors in Rydberg qudit systems.
\item As an extra theoretical result, we apply the techniques of Wurtz and Love~\cite{Wurtz2020} to the case of correlation clustering on 3-regular graphs, showing that for a slightly modified QAOA at $p=1$ the approximation ratio is at least $0.6367$. From numerical results obtained through local optimization, we conjecture that a tighter bound of this approximation ratio would be $0.6699$.
\end{itemize} 

\paragraph*{Related work}
Independently of the current work, other studies have considered applying QAOA to multi-cut versions of MAX-$k$-CUT~\cite{fuchs2021efficient}, and on MAX k-VERTEX COVER~\cite{cook2020quantum,bartschi2020grover}, which are problems with related properties.
A key difference between correlation clustering and the other studied problems is the explicit presence of positive / negative edges in correlation clustering, and the idea that for correlation clustering the number of clusters is not determined yet as part of the input.
Additionally, we optimize the creation of the QAOA formulation, ending up with a native qudit implementation, as well as taking a full stack approach: we study all the steps from the problem towards the implementation on a realistic near-term quantum device.

~\\
The paper is organized as follows. In Sec.~\ref{sec:Recap} we briefly recap the QAOA. In Sec.~\ref{sec:Int} we introduce the problem of correlation clustering and its implementation as QAOA. We then describe various strategies to improve the algorithm's performance
in Sec.~\ref{sec:IS}, which we study in Sec.~\ref{sec:P}. Next, in Sec.~\ref{sec:RQC_BuildingBlocks} we discuss the experimental building blocks of the qudit Rydberg quantum computer. We proceed with the processor design and the associated gate count and comparison to qubits in Sec.~\ref{sec:PD}. Finally, we discuss how the errors affect the algorithm in Sec.~\ref{sec:Errors}, and we conclude and discuss open questions in Sec.~\ref{sec:Conclusions}




\section{\label{sec:Recap} Recap of QAOA}
In this section we briefly review the quantum approximate optimization algorithm (QAOA)~\cite{Farhi2014}. Consider some combinatorial optimization problem with objective function $C:x\rightarrow \mathbb{R}$ acting on $n$-bit strings $x\in \{0,1\}^n$, domain $\mathcal{D} \subseteq \{0,1\}^n$, and objective
\begin{align}
     \max_{x \in \mathcal{D}} C(x).
\end{align}
In maximization, an approximate optimization algorithm aims to find a string $x'$ that achieves a desired approximation ratio $\alpha$, i.e.
\begin{equation}
    \frac{C(x')}{C^*}\geq \alpha,
\end{equation}
where $C^* = \max_{x \in \mathcal{D}} C(x)$.
In QAOA, such combinatorial optimization problems are encoded into a cost Hamiltonian $H_C$, a mixing Hamiltonian $H_M$ and some initial quantum state $\ket{\psi_0}$. The cost Hamiltonian is diagonal in the computational basis by design, and represents $C$ if its eigenvalues satisfy
\begin{align}
    H_C \ket{x} = C(x) \ket{x} \text{ for all } x \in  \{0,1\}^n.
\end{align}
The mixing Hamiltonian $H_M$ depends on $\mathcal{D}$ and its structure~\cite{Hadfield2019}, and is in the unconstrained case (i.e. when $\mathcal{D}=\{0,1\}^n$) usually taken to be the transverse field Hamiltonian $H_M = \sum_{j} X_j$. Constraints (i.e. when $\mathcal{D}\subset \{0,1\}^n$) can be incorporated directly into the mixing Hamiltonian or are added as a penalty function in the cost Hamiltonian. The initial quantum state $\ket{\psi_0}$ is usually taken as the uniform superposition over all possible states in the domain. $\text{QAOA}_p$, parametrized in $\gamma=(\gamma_0,\gamma_1,\dots,\gamma_{p-1}),\beta=(\beta_0,\beta_1,\dots,\beta_{p-1})$, refers to a level-$p$ QAOA circuit that applies $p$ steps of alternating time evolutions of the cost and mixing Hamiltonians on the initial state. At step $k$, the unitaries of the time evolutions are given by
\begin{align}
    U_C(\gamma_k) = e^{-i \gamma_k H_C }, \label{eq:UC} \\
    U_M(\beta_k) = e^{-i \beta_k H_M }. \label{eq:UM}
\end{align}
So the final state $\ket{\gamma,\beta}$ of $\text{QAOA}_p$ is given by 
\begin{align}
    \ket{\gamma,\beta} = \prod_{k=0}^{p-1} U_M(\beta_k) U_C(\gamma_k) \ket{\psi_0}.
\end{align}
The expectation value $ F_p(\gamma,\beta)$ of the cost Hamiltonian for state $\ket{\gamma,\beta}$ is given by
\begin{align}
    F_p(\gamma,\beta) = 
    \bra{\gamma,\beta}H_C\ket{\gamma,\beta},
    \label{eq:Fp}
\end{align} and can be statistically estimated by taking samples of $\ket{\gamma,\beta}$. The achieved approximation ratio (in expectation) of $\text{QAOA}_p$ is then
\begin{equation}
    \alpha = \frac{F_p(\gamma,\beta)}{C^*}.
\end{equation}
The parameter combinations of $\gamma,\beta$ are usually found through a classical optimization procedure that uses \eqref{eq:Fp} as a black-box function to be maximized. The QAOA framework as has been described so far, with randomized initial points and without any other improvement strategies, will be referred to as the \textit{vanilla QAOA} in the rest of this work.


\section{\label{sec:Int} Correlation clustering}
Generally, the objective of clustering problems is to group elements into a family of subsets, named clusters, such that the elements within a cluster are more similar to one another than elements in different clusters.
In case of the correlation clustering problem, we would like to cluster without specifying the number of clusters in advance based only on pairwise relations.
The problem was introduced by Bansal et al.~\cite{Bansal2004} to the theoretical computer science community and has applications amongst others in social psychology, statistical mechanics and biological networks. 

Instances of correlation clustering problems are commonly represented as a graph problem, where the nodes are the elements to be grouped in clusters and edge weights represent similarities between these elements. Correlation clustering is then formally defined in the following way: let $G(V,E)$ be an undirected graph, where $V,E$, denotes the sets of nodes and edges, respectively. Let $N$ be the total amount of nodes, i.e. $N = |V|$. Every edge $(u,v)\in E$ is labelled either `+' or `--', depending on whether the elements are similar or dissimilar, respectively. This is the unweighted variant of correlation clustering. Additionally, one can also consider the weighted variant: edges $(u,v)$ carry weights $w_{(u,v)}\in \mathbb{R}^+$ describing an additional measurement of similarity or dissimilarity. There are two complementary problems to correlation clustering. MAXAGREE aims to maximize the number of agreements, defined as the number of `+' edges inside clusters plus the number of `--' edges across clusters. In MINDISAGREE one wants to minimize the number of disagreements: the number of `+' edges across different clusters plus the number of `--' edges inside clusters. The decision versions of MAXAGREE and MINDISAGREE are identical and known to be NP-complete~\cite{Bansal2004}. However, they differ in the approximation setting. MAXAGREE on general graphs is APX-hard: to be precise, it has been shown that for every $\epsilon > 0$, it is NP-hard to approximate both the weighted~\cite{Charikar2005} and unweighted~\cite{Tan2007} versions of MAXAGREE within a factor of $79/80 + \epsilon$, and the best classical algorithm has approximation ratio $\alpha =0.7666$~\cite{Swamy2004} via semi-definite programming (SDP) with rounding techniques. In the remainder of the text this specific value of $\alpha$ will be refered to as the \textit{Swamy bound}, named after the author of Ref~\cite{Swamy2004}. If the graph is complete  and unweighted the problem becomes significantly easier although still NP-hard. For  complete graphs, Bansal et al.~provided a polynomial time approximation scheme (PTAS)~\cite{Bansal2004}. MINDISAGREE is APX-hard on \textit{both} complete~\cite{Bansal2004} and general graphs~\cite{Charikar2005}, and the best approximation algorithm achieves only an approximation logarithmic in the input size~\cite{Emanuel2003}. When the correlation clustering problem is restricted to a fixed number of clusters $k$, it also remains NP-hard when $k\geq2$, albeit with existing PTASs for both the maximization as well as the minimization variant \cite{Giotis2006}.

Other algorithmic methods that have been proposed to solve correlation clustering include integer linear programming (ILP)  and heuristics: examples are greedy methods~\cite{Soon2001,Ng2002,Ailon2008}, local search methods~\cite{Gionis2007,Goder2008}, and large move making algorithms~\cite{Bagon2011}. In practice, the ILP-approach can solve to about 200 nodes due to a scaling of $O(N^3)$ in the amount of constraints in its formulation. The SDP relaxation method has better scaling in its constraints ($O(N^2)$), with the best SDP solvers handling up to about several thousands of nodes ~\cite{Elsner2009}. Heuristics have been shown to be able to handle problems consisting over 100k nodes~\cite{Bagon2011}. In this work, we will focus on unweighted correlation clustering in complete and Erd\H{o}s–R\'enyi graphs. 

A specific real-world application of the correlation clustering problem can be found in the sub-task of distinguishing between persons within multi-person pose estimation. For example, decomposing different persons is necessary in human-robot collaborations where humans fulfil tasks collaboratively with robots and a robot has to differentiate between the position of different human collaborators.
In the approach presented in~\cite{InsafPAAS16}, a Convolutional Neural Network (CNN) computes candidates for the location of different body parts of different persons within a given picture and also pairwise terms how these candidates relate to each other.
Leaving out the specific labelling of body parts of each person, the detected body part candidates from the CNN can be represented as nodes and the computed pairwise terms as edge weights of a graph and thus solving the correlation clustering problem of this graph corresponds to a decomposition of the persons in a picture (cf.\ illustration in Fig.~\ref{fig:example_cc}).
From industrial practice we know that the objective function within the correlation clustering problem can be very sensitive to the assignments of parameters and sometimes even a sufficiently good enough local minimum cannot be found within reasonable time by using classical algorithms.
\begin{figure}
    \centering
    \includegraphics[width = \linewidth]{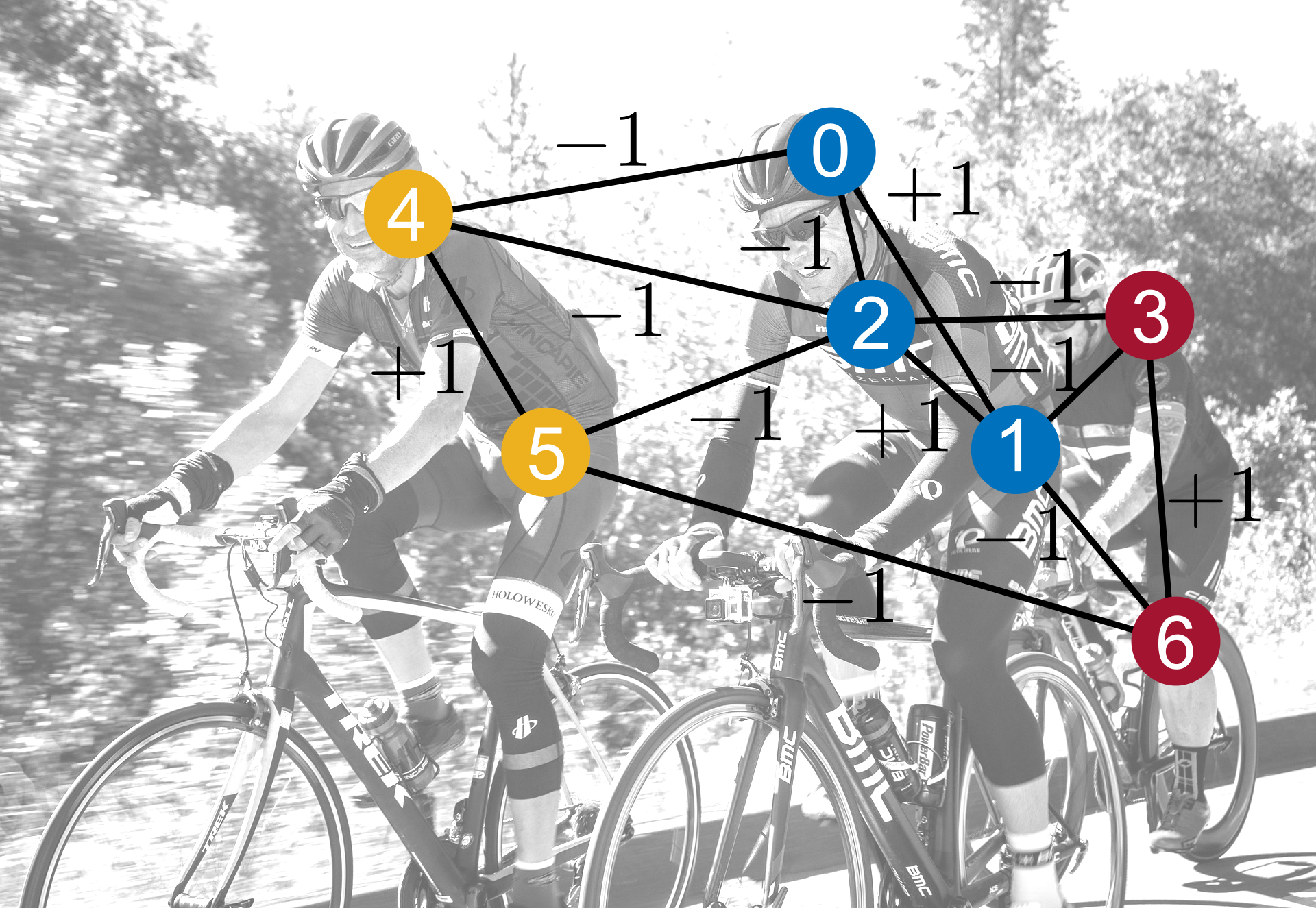}
   \caption{Overlay of a correlation clustering instance on a real-world picture to illustrate the occurrence of correlation clustering within the problem of decomposing persons in images. The graph contains $N=7$ nodes with weights $w_{(u,v)} = \pm 1$ while the colours indicate an optimal clustering with objective function value of 10 in the MAXAGREE setting. Note how in general the graph does not need to have a perfect (unfrustrated) clustering, i.e. where all allocated clusters match the weights, as is in this example not the case for nodes `0' and `2'.}
    \label{fig:example_cc}
\end{figure}

\subsection{\label{sec:HF} Hamiltonian formulation for qudit systems}
Let $G=(V,E)$ be a graph with $N=|V|$ nodes that specifies the input to some correlation clustering problem instance. Every edge $(u,v)\in E$ has a weight $w_{(u,v)} \in \{-1,+1\}$, representing the `+' and `--' relationships in unweighted correlation clustering. We assume that we have access to a qudit quantum system of $N$ qudits consisting of $d$ levels, such that every node is described by a single qudit state $\ket{c_u}$ meaning that node $u$ is put in (a superposition of) cluster(s) $c$
\footnote{
Here, we would like to stress that even though $d$ might be limited, this will not be too detrimental to the QAOA’s performance in practice, even when the optimal solution would require more than $d$ clusters. 
This was already shown for the Swamy algorithm, which always produces a solution that uses at most 6 clusters and is still able to achieve a high expected approximation ratio \cite{Swamy2004}. Consider for example the all-negative weights graph, which would be an instance that requires the maximal amount of clusters to be used, which equals the number of the graph vertices, $d=N$. In this case, one would still be able to satisfy at least $d-1$ out of $d$ edges and thus is still able to achieve an approximation ratio larger than $(d-1)/d$. For the experimental example discussed in Sec.~\ref{sec:RQC_BuildingBlocks}, we have $d=10$ and thus the minimum approximation ratio would be $9/10=0.9$.
}.
We write $[d]$ for the set $\{0,\dots,d-1\}$. We define a two-body interaction $V^d_{(u,v)}$ that acts on a two qudit sub-space according to 
\begin{align}
   V^d_{(u,v)}= \sum_{i \neq j \in [d]} \ket{i_u}\ket{j_v}\bra{i_u}\bra{j_v} - \sum_{i \in [d]} \ket{i_u}\ket{i_v} \bra{i_u}\bra{i_v},
    \label{eq:defV}
\end{align}
which is a $d^2\cross d^2$ matrix with eigenvalues of -1 (+1) for nodes that are put in the same (different) clusters. Our full cost Hamiltonian is obtained by summing over all edges taking the product of the edge weight and~\eqref{eq:defV},
\begin{align}
    H_C = \sum_{(u,v)\in E} w_{(u,v)} V^d_{(u,v)}.
    \label{eq:HC}
\end{align}
 To follow the convention in physics, we aim to minimize this cost Hamiltonian -- note that this is equivalent to maximizing the classical cost function. We comment that this encoding is similar to the binary encoding for MAX-$k$-CUT recently proposed in Refs.~\cite{fuchs2021efficient,Bravyi2020}.
 A Hamiltonian that can mix over the single qudit subspace was given in the work by Hadfield et al.~\cite{Hadfield2019}, where the following single-qudit mixing Hamiltonian is proposed
\begin{equation}
    h_M (r) = \sum_{i=1}^{r} \left((\Sigma^x)^i+(\Sigma^{x\dagger})^i   \right),
    \label{eq:ring_mixer}
\end{equation}
where $r\in\{1,\dots,d-1\}$, a parameter determining the connectivity of the mixer and $\Sigma^x$ is the generalized Pauli X-operator, given by
\begin{equation}
\begin{gathered}
\Sigma^x  = 
\begin{pmatrix}
		0 &  & & 1\\
		1 & &  \\
		& \ddots & &  \\
		& & 1 & 0
		\end{pmatrix}.
\end{gathered}
\label{eq:Sigma_x}
\end{equation}
One observes that for $r=1$, the single-qudit mixer is equal to 
\begin{equation}
\begin{gathered}
h_M (r=1) = 
\begin{pmatrix}
		0 & 1 & & & 1\\
		1 & &\ddots&   \\
		& \ddots &  &  \ddots &  \\
		 & & \ddots  &  & 1\\
		1 & & & 1 & 0 
		\end{pmatrix}
\end{gathered}
\label{eq:mixer_matrix}
\end{equation}
such that every level is connected to its nearest neighbours, including periodic boundary conditions. The full mixing Hamiltonian is then
\begin{align}
    H_M = \sum_{u \in V} h_M.
    \label{eq:qudit_mixer}
\end{align}
We can pick any value of $r\in \{1,\dots,d-1\}$, where the special cases at the boundary are called the single-qudit ring mixer for $r=1$ and the fully-connected mixer for $r = d -1$. We take the superposition of all qudit computational basis states as our initial state, i.e.
\begin{align}
    \ket{\psi_0} =  \frac{1}{\sqrt{d^{N}}} \sum_{z\in [d]^{N}} \ket{z}.
    \label{eq:ML_is}
\end{align}\\
Note how the cost Hamiltonian formulation~\eqref{eq:HC} is not equivalent to correlation clustering in the MAXAGREE setting: instead of counting just the agreements we count the number of agreements \textit{minus} the disagreements. However, since every clustering of two nodes (for which an edge exists), needs to be in either agreement or disagreement with the corresponding weight, the sum of the agreements and disagreements is equal to the number of edges in the unweighted setting. Therefore, for a correlation clustering problem with optimum value $C^*$ in the MAXAGREE setting the approximation ratio of this QAOA formulation is equal to
\begin{equation}
     \frac{F_p(\gamma,\beta)+|E|}{2 C^*}.
\end{equation}
For the numerical simulations in the algorithmic sections of this work, we implement the initial state by generalized Hadamard operations and assume that the cost and mixing unitaries are elementary operations native to our system. We adopt the Cirq framework~\cite{Cirq2021}, since it supports qudit systems, with custom gate operations to match our established formulation. 


\section{\label{sec:IS} Improvement strategies}
We will be interested in performance at low-depth and hence we only consider $r=1$ in Eq.\eqref{eq:ring_mixer} for the simulations in sections~\ref{sec:IS} and~\ref{sec:P}. Through numerical evaluation of the vanilla QAOA, we found that the following strategies considerably improved the QAOA's performance for our problem:\\
\paragraph*{Choice of the classical optimizer}
There is no such thing as a one-size-fits-all classical optimizer that performs well for all QAOA problems: performance varies amongst different problems and can greatly differ for different hyper-parameter settings~\cite{Lavrijsen2020}. We decided to compare different classical optimizers, found in the scikit-quant~\cite{Lavrijsen2020} and SciPy Python packages~\cite{Scipy}, using their off-the-shelf hyper-parameter settings. In a small study, the best performance was obtained by using BOBYQA. The results of the full comparison can be found in Appendix~\ref{app:optstudy}.\\
\paragraph*{Restarts.}
Local optimizers can greatly benefit from restarts (i.e.\ re-running the same algorithmic procedure with the same parameter settings) since they are sensitive to the quality of initial points. But even with fixed initial points, due to stochastic elements in the optimizers' procedures as well as being introduced from sampling (or gate errors), the algorithms can be made more robust by incorporating restarts. The work of Shaydulin et al.~shows how multi-start methods can improve QAOA~\cite{Shaydulin20192}. In our numerical simulations we set the number of restarts to 5.\\
\paragraph*{Optimised initial points.}
Numerical and analytical works have indicated that QAOA's parameters $\gamma,\beta$ are concentrated in parameter space for instances that belong to `similar classes'~\cite{Brandao2018,Streif20192,Akshay2021}. In practice, this means that we can select a single instance from a class, work very hard to find optimal parameters through variational optimization or from analytical arguments, and store these values to be used as initial points when solving other instances. For every considered $N,d$, we use the all-negative-weights instance to obtain the parameter values to be used as initial points. Numerical results of a study at $p=1$ can be found in Appendix~\ref{app:ipstudy}.\\
\paragraph*{Looping over clusters.}
Initial experiments showed that the vanilla QAOA performed well on instances for which the optimal solution required many clusters -- this means that the optimal cluster number was close to the available number of clusters $d$ -- but not so when the number of clusters was low. Since our Hamiltonian formulation allows for $d$ to be varied, we can iterate the algorithm over all possible maximum number of clusters, i.e. $d\in\{1,\dots,N\}$, and keep the best result obtained over all iterations.\\
\begin{figure}
    \centering
    \includegraphics[width = \linewidth]{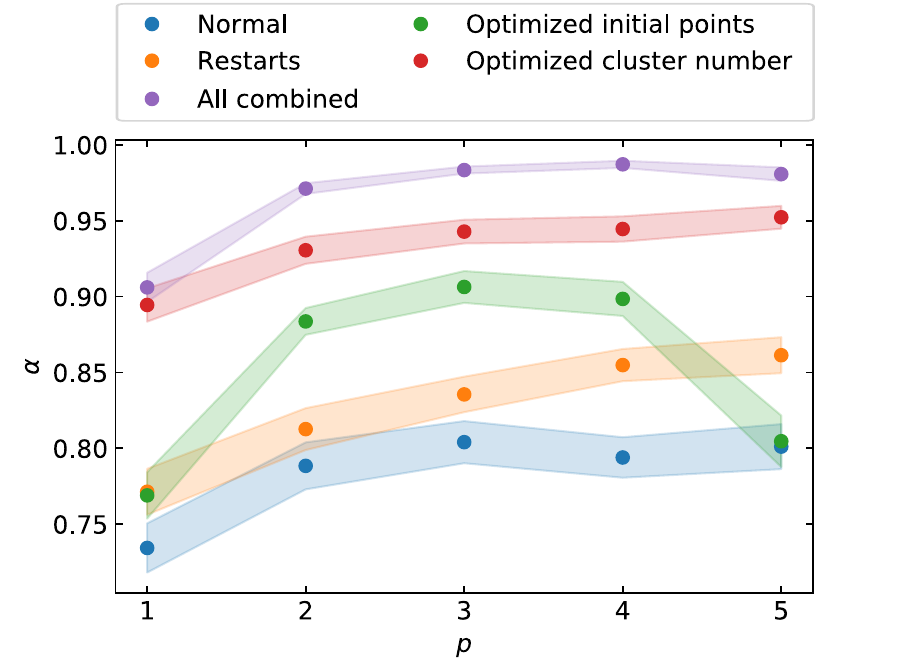}
    \caption{Approximation ratios $\alpha$ for different improvements added to the vanilla QAOA algorithm applied to the $N=4$ complete graph data set. The blue plot (label `Normal') represents the performance without any of the improvement strategies as suggested in the text. In all cases COBYLA was used as the classical optimizer. The filled circles data points indicate the average value over all 50 instances and the shaded area represents the error in the mean.}
    \label{fig:Improvement_strategies}
\end{figure}
\linebreak
Fig.~\ref{fig:Improvement_strategies} shows the numerical results of the performance on a data set of 50 correlation clustering instances on complete graphs with weights randomly picked out of $\{-1,+1\}$. In the creation of the data set, we swept the probability for giving an edge weight `$+1$' from 0 to 1 uniformly in order to have a good representation of all possible weight configurations. We observe that looping over cluster numbers has the largest contribution to the improvement, followed by the optimized initial points, which work particularly well for intermediate values of $ 2\leq p \leq4$. This could be due to the fact that for $p=1$ the optimizer is able to find good points even with bad initial points, whilst for $p=5$ the distance from the initial point and the optimal point might be larger in parameter space, which results in the initial point becoming effectively random again. The restarts are expected to have less of an impact when used in conjunction with good initial points. Examples of other variants and strategies to QAOA that one could consider, but were not used in this work, are RQAOA~\cite{Bravyi2020}, ADAPT-QAOA~\cite{Zhu2020} and parameter initialization heuristics~\cite{Zhou2020}. 

\section{\label{sec:P} Performance}
We consider correlation clustering data sets consisting of complete graphs as well as Erd\H{o}s–R\'enyi graphs: these are random graphs with a fixed amount of nodes $N$ but for which the edges are created according to some edge creation probability. In the creation of our Erd\H{o}s–R\'enyi data sets we use an edge creation probability of $0.5$, and additionally have the criteria that every instance needs to have at least one edge. For both graph types we again sweep the edge weight probability of giving an edge weight `$+1$' from 0 to 1 to create 50 random instances. 
\begin{figure}
    \centering
    \includegraphics[width=\linewidth]{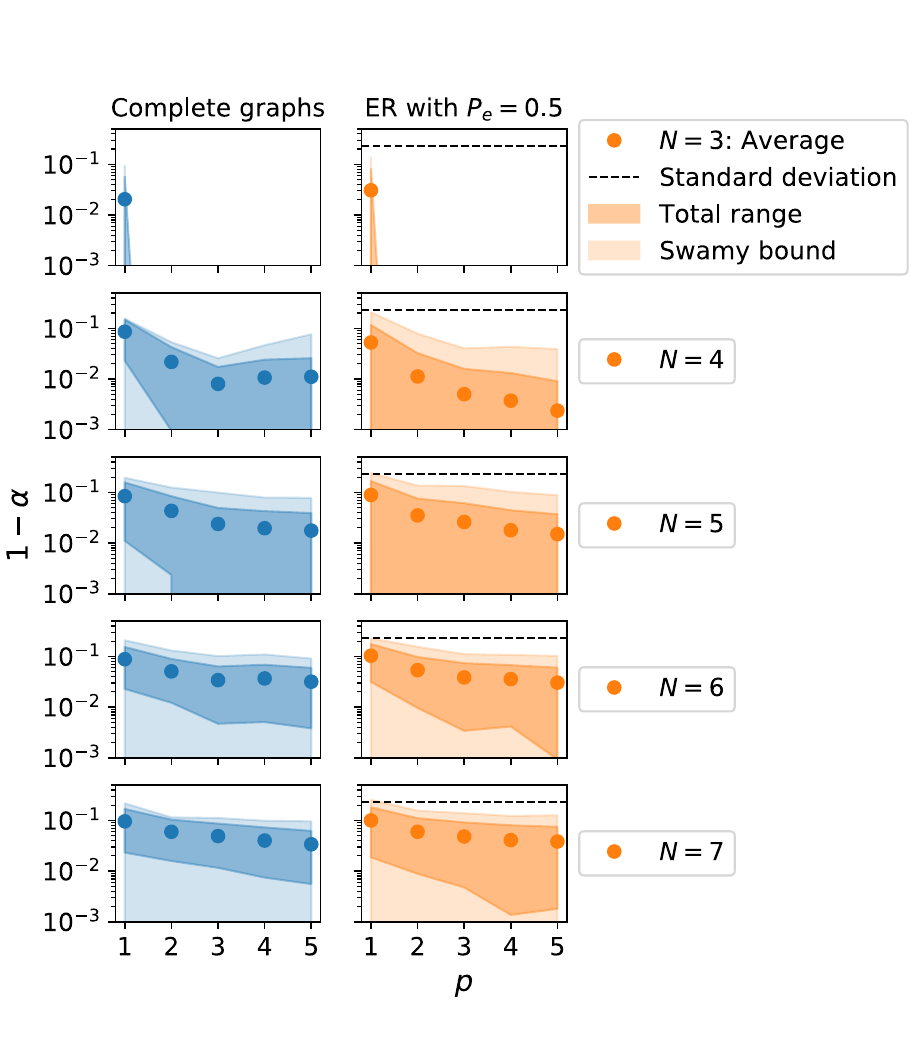}
    \caption{Performance of QAOA, measured in obtained 1 minus the approximation ratio (1-$\alpha$), on the data sets of complete graphs (blue) and Erdős–Rényi graphs (orange) with edge creation probabilities $P_e=0.5$ as a function of the number of nodes $N$ and $p$. The Swamy bound is added only for the Erdős–Rényi graphs as a PTAS exists for complete graphs. To improve readability, artificial continuity is added between the discrete data for standard deviation (dark shaded), total range (light shaded) and the Swamy bound (dashed black line).
    }
    \label{fig:PAPER_1}
\end{figure}

Fig.~\ref{fig:PAPER_1} shows the  average values, standard deviation and total range (defined as the worst and best performance over all instances in a data set of graphs of fixed size $N$) of the achieved approximation ratios on instances of our data sets as functions of $N$ and QAOA depth $p$, achieved by the QAOA that adopts all strategies as listed in section~\ref{sec:IS}. We observe that the worst performance is slightly better for complete graphs than it is for Erdős–Rényi graphs, which is to be expected as the complete graph problem instances have more structure and are also more easily solved classically~\cite{Bansal2004}. Furthermore, we find a slightly more profound $p$-dependence, as reflected for instance in the amplitude of the standard deviation with increasing $p$, for Erdős–Rényi graphs as compared to the complete graphs. This is not surprising, as for the complete graphs the QAOA at $p=2$ is already able to `see the whole graph': for every edge any other edge is at most $p$ edges away. The average performance is comparable amongst different graph types. For $p=1$ we observe that the worst performance on instances from our bench-marking data sets becomes comparable to the Swamy bound of $0.7666$ for $N\geq 5$, but for $p=2$ we have that the algorithm performs better than this bound on all instances in the data set
\begin{figure}
    \centering
    \includegraphics[width=\linewidth]{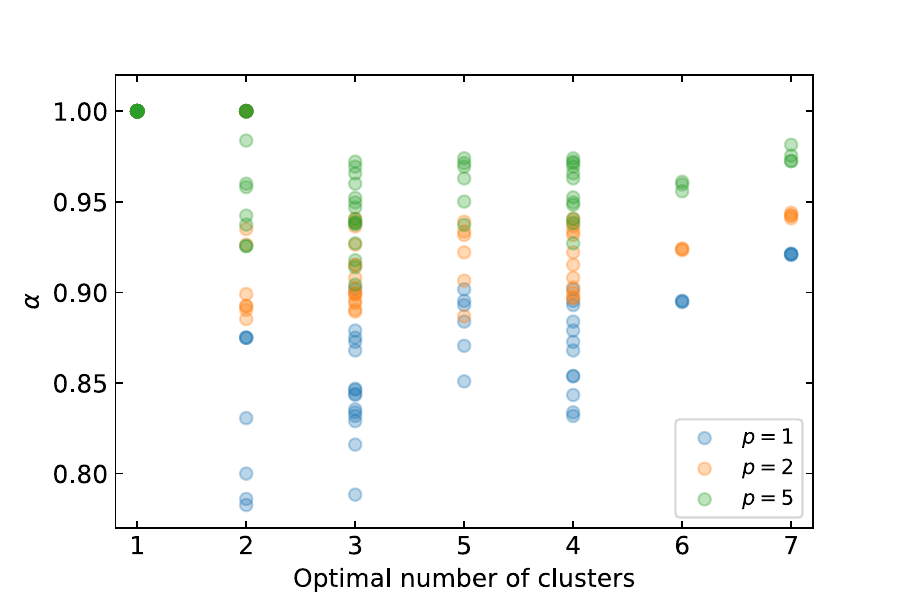}
    \caption{Scatter plot of the achieved approximation ratios $\alpha$ on the $N=7$ complete graph data set as a function of the optimal number of clusters, which was obtained through brute force search. In the case when multiple different cluster numbers were optimal, the data points are plotted for all of these values.}
    \label{fig:PAPER_3}
\end{figure}

In order to investigate instance-dependence of the performance, Fig.~\ref{fig:PAPER_3} shows a scatter plot of the performance on individual instances in the complete graph data set of $N=7$ for different values of $p$. We find that the algorithm at low $p$ has the most difficulty with instances that require a low number of clusters (except the singleton cluster case, which is trivial for $d=1$), and as $p$ increases the most difficult instances seem to have optimal cluster numbers in the middle between the singleton and all-different clusters.
\begin{figure}
    \centering
    \includegraphics[width=\linewidth]{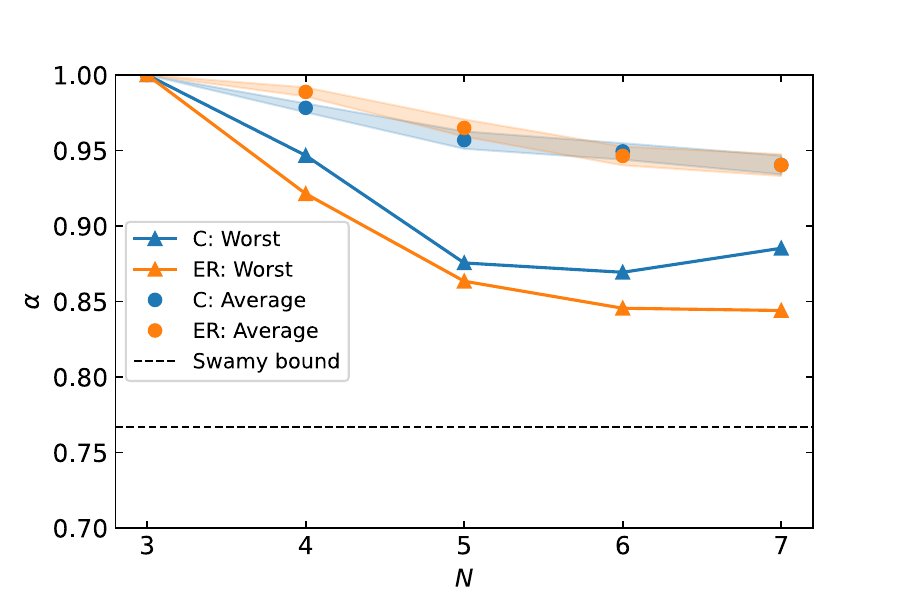}
    \caption{Average (filled circles) and worst (triangles) approximation ratios $\alpha$ as a function of $N$ for $p=2$ for the complete graph (labelled 'C') and Erdős–Rényi (labelled 'ER')  data sets. The dashed line is again the Swamy bound. The shaded area represents the error in the mean.}
    \label{fig:PAPER_2}
\end{figure}

Next, Fig.~\ref{fig:PAPER_2} shows the average and worst approximation ratios a function of $N$ giving an indication of the scalability. We observe that the decrease of the worst approximation ratio in $N$ seems to slow down considerably for $N \geq 5$ (for complete graphs it even improves). It should be noted though that it is difficult to make too definitive statements here: we were limited to studying graphs up to $N=7$ and have no guarantee that we used optimal parameters. However, we can show that the approximation ratio in the limit of large $N$ can in fact become independent of $N$. We propose (conjecture) that in this limit $\text{QAOA}_1$, looping over $d=1,2,3,4$, has a performance guarantee of $0.6367$ ($0.6699$) on 3-regular graphs. The full derivation of this bound can be found in Appendix~\ref{app:app_ratio_bound}, and follows a similar method as the one proposed by Wurtz and Love~\cite{Wurtz2020}. However, it is not yet clear how this bound changes as a function of the graph's degree. For MAXCUT~\cite{wang2018quantum} and Maximum Independent Set (MIS)~\cite{Farhi2020} it has already been shown that the approximation ratio at $p=1$ decreases as the degree of the graph increases. Overall, our results indicate that QAOA, albeit with some added heuristic strategies, might achieve a competitive approximation ratios with respect to the best classical approximation algorithm in solving correlation clustering problems of low-degree graphs at low circuit depth.

\begin{figure*}[t!]
	\centering
	\includegraphics[width=\textwidth]{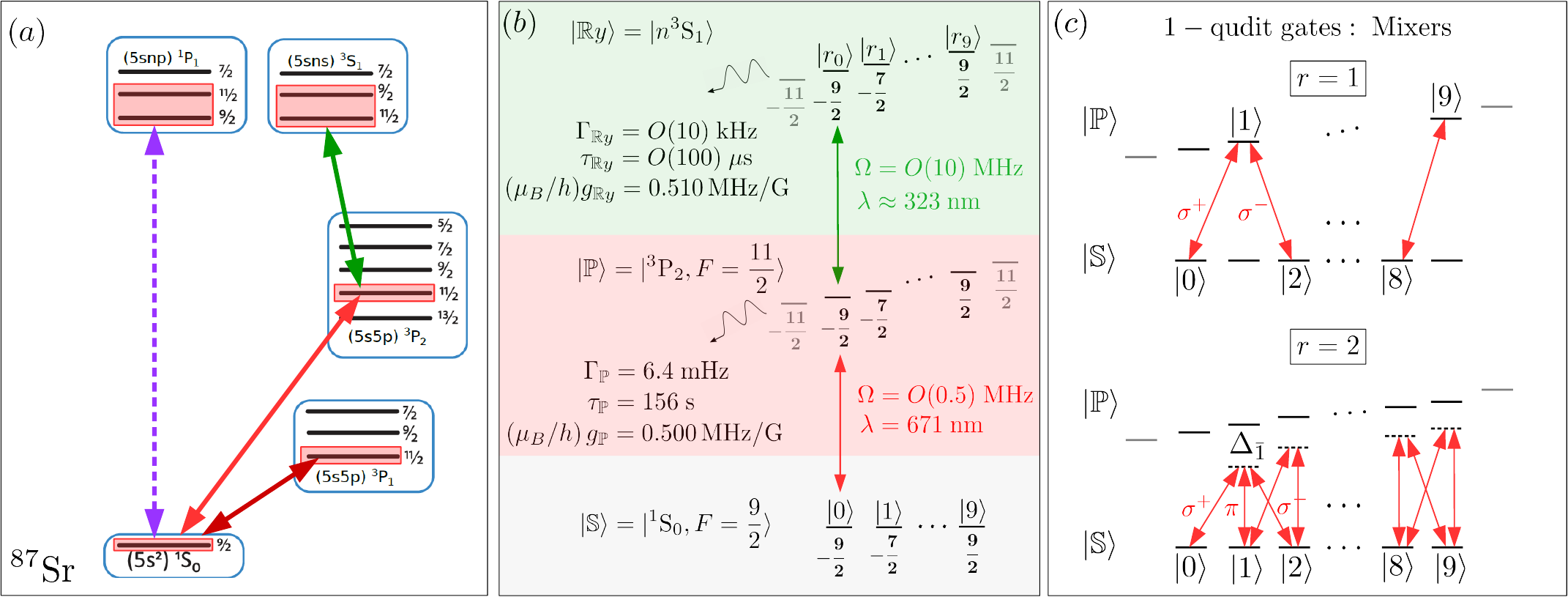}
	\caption{
	{\bf (a)} Relevant level scheme of ${}^{87}$Sr.	Proposed qudit states are realized by the ground state manifold $\ket{\mathbb S} = \ket{{}^1{\rm S}_0,F=9/2}$. The two-qudit gates are realized by exciting the states from $\ket{\mathbb S}$ to a Rydberg manifold $\ket{\mathbb Ry} = \ket{{\rm n}^3{\rm S}_1}$ through an intermediate state from the $\ket{\mathbb P} = \ket{{}^3{\rm P}_2, F=11/2}$ manifold (red and green arrows). The dashed violet arrow indicates an alternative possibility to excite the ground state to a Rydberg manifold $\ket{{\rm n}^1{\rm P}_1}$, using a single photon transition. The dark red arrow indicates the transition to the $\ket{{}^3{\rm P}_1, F=11/2}$ manifold used for measurement of the quantum state.
	{\bf (b)} Parameters of the manifolds $\ket{\mathbb S}, \ket{\mathbb P}, \ket{\mathbb Ry}$ relevant for the qudit operations: transition wavelengths $\lambda$, typical Rabi frequencies $\Omega$, decay rates $\Gamma$ and the associated lifetimes $\tau$ from the excited $\ket{\mathbb P}, \ket{\mathbb Ry}$ manifolds and the Land\'{e} $g$-factor quantifying the Zeeman splitting of the magnetic sublevels. The values of $\Gamma_{\mathbb P}, \tau_{\mathbb P}$ are taken from~\cite{Porsev_2004_PRA}.
	The $\mathbb P$-manifold states are used to realize the single qudit gates such as the mixing unitaries, which are shown in {\bf (c)} for $r=1$ and $r=2$, see text for details.
	}	
	\label{fig:level_scheme}
\end{figure*}

\section{\label{sec:RQC_BuildingBlocks} Experimental building blocks}

In this section we describe the proposed implementation of the qudits. While most current experimental efforts focus on the use of qubits, neutral atoms are a natural platform for qudits and we specifically focus on the example of fermionic strontium, $^{87}{\rm Sr}$~\cite{DeSalvo_2010_PRL, Stellmer_2013_PRA, stellmer2014degenerate}. The reason is that it possesses a nuclear spin $I=9/2$ that is decoupled from the electronic spin, which features $d_{\rm max}=2I+1=10$ hyperfine states in the ground state manifold, which are insensitive to electric and magnetic field fluctuations. Moreover, one can make use of the long-lived excited states from the ${}^3{\rm P}_J$ manifold, which has been exploited in a recent experiment to create a Bell state with fidelity reaching 99\%~\cite{Madjarov_2020_NatPhys}. The analysis presented below can be adapted to the analogous isotope ${}^{173}$Yb of fermionic ytterbium~\cite{Fukuhara_2007_PRL, Taie_2010_PRL, sugawa2013ultracold}, where however the maximum available $d_{\rm max}=6$ in the ground state manifold is smaller compared to $d_{\rm max}=10$ of ${}^{87}{\rm Sr}$. In the following, we refer to single and two qudit gates as 1-gates and 2-gates, respectively.\\
%
~\\

\noindent
\emph{Qudit manifold.}
The relevant level scheme of ${}^{87}{\rm Sr}$ is sketched in Fig.~\ref{fig:level_scheme}a. As stated above, the ground state manifold, which we denote with a slight abuse of notation as $\ket{{\mathbb S}} \equiv \ket{{}^1{\rm S}_0, \; F=\frac{9}{2}}$, consists of $d=2 F + 1 = 10$ $m_F$-states, $m_F \in \{-9/2, \ldots, 9/2 \}$. We also denote $\ket{\mathbb P} \equiv \ket{{}^3{\rm P}_2, F=\frac{11}{2}}$ the excited state manifold, which we will use to implement the 1- and 2-gates~\cite{Onishchenko_2019_PRA} (here we choose the ${}^3{\rm P}_2$ manifold in particular due to its long lifetime, which allows for a resonant excitation to the Rydberg state, cf. below).
The optical tweezers providing the trapping potentials are typically realized with light of wavelength $\lambda_{\rm tweezer}$ that is red-detuned from the dominant $\ket{\mathbb S} - \ket{{}^1{\rm P}_1}$ trapping transition (not shown in the Fig.~\ref{fig:level_scheme}a).
The choice of the ${\mathbb P}$-manifold is motivated by the fact that, unlike the other possible choices such as $\ket{{}^3{\rm P}_2, F=\frac{7}{2}}$ or $\ket{{}^3{\rm P}_2, F=\frac{9}{2}}$, it possesses a so-called \emph{magic} wavelength $\lambda_{\rm tweezer} \approx 900 \, {\rm nm}$, for which the transition frequencies $\ket{{\mathbb S},m_F} - \ket{{\mathbb P},m_{F'}}$ are approximately independent of the intensity of the tweezer light for all $m_{F}, m_{F'}$, which ensures a position independent addressing frequency of the individual $m_{F}$ states~\cite{Urech1, Safronova, Deutsch_1998_PRA}.
The actual addressing relies on the Zeeman splitting of the ${\mathbb P}$-manifold and has been experimentally demonstrated using ${}^{173}{\rm Yb}$~\cite{Taie_2010_PRL}. Applying moderate values of the magnetic field $B$ results in linear Zeeman splitting with an energy shift between the adjacent $m_{F'}$ states of $\mu_B g B/h$, where $\mu_B$ and $h$ is the Bohr magneton and the Planck constant respectively and $g$ is the Land\'{e} $g$-factor. For the $\mathbb P$-manifold, $g \approx 0.36$ and $\mu_B g/h \approx 0.5\,{\rm MHz}/{\rm G}$~\cite{Boyd_2007_PRA}. This provides a splitting of $\approx 50 \; {\rm MHz}$ between adjacent $m_{F'}$ states for a magnetic field amplitude of $100 \; {\rm G}$, allowing for both resonant and off-resonant addressing as we now discuss.

~\\
\noindent
\emph{Rydberg states.} 
The 2-gates are implemented via the Rydberg blockade provided by the interaction energy $V$, which for the density-density interaction between a pair of atoms separated by a distance $R$ typically corresponds to a Van der Waals type $V(R) = C_6/R^6$, where $C_6$ is so-called ``$C_6$'' (or Van der Waals) coefficient~\cite{gallagher2005rydberg,saffman2016quantum,Saffman_2010_RMP,browaeys2020many}.

The atom can be excited from the $\mathbb S$-manifold through a two-photon transition via the $\mathbb P$-manifold to a Rydberg state $\ket{\mathbb Ry} = \ket{n {}^3{\rm S}_1}$ (red and green arrows in Fig~\ref{fig:level_scheme}a,b). In principle one could also use a direct one-photon transition to the $\ket{n {}^1{\rm P}_1}$ Rydberg state (dashed purple line in Fig.~\ref{fig:level_scheme}a). However, due to current technology limitations, such as lack of lasers of sufficient power (and the associated optical elements) for the $\ket{\mathbb S}-\ket{n {}^1{\rm P}_1}$ transition wavelength of $\sim 220 \, {\rm nm}$, we solely focus on the $\ket{\mathbb S} - \ket{\mathbb P} - \ket{\mathbb Ry}$ scheme, cf. Fig.~\ref{fig:level_scheme}b.
This allows for $m_F$ specific addressability of the Rydberg states as well.
The two-photon transition is typically operated off-resonantly but given the extremely long lifetime of the $\mathbb P$-manifold ($\tau_{\mathbb P} = 156\,s$ in the absence of the tweezer light) and to further reduce the timescales needed for operation we consider a resonant two-step process: the chosen $m_F$ state is transferred as $\ket{{\mathbb S}, m_F} \, \rightarrow \, \ket{{\mathbb P},m_{F'}} \, \rightarrow \, \ket{\mathbb Ry}$ using two consecutive $\pi$-pulses or using stimulated rapid adiabatic passage. The properties of the relevant manifolds are summarized in Fig.~\ref{fig:level_scheme}b.

In the following, we refer to $\ket{\mathbb S}$ as the qudit manifold. We note that owing to the long lifetime of the $\ket{\mathbb P}$-manifold, one could in principle choose it as the basis for the qudit states \footnote{We note that analogous approach has been adopted in Ref.~\cite{Madjarov_2020_NatPhys}, where the ${}^3P_0$ manifold has been used for a realization of a qubit state.} and the subsequent analysis can be readily adapted to this alternative choice.

Having identified the suitable level structure we now proceed with the design of the 1- and 2-qudit gates.

\subsection{1-gates}
\label{sec:RQC_1gates}

We consider the implementation of a qudit 1-gate by coupling the qudit level $\ket{\ell}$ to level $\ket{\ell'}$ by means of laser fields of Rabi frequency $\Omega_{\ell ,\ell'} \equiv \Omega c_{\ell,\ell'}$, where $\Omega \in {\mathbb R}_+$ is the Rabi frequency amplitude and $c_{\ell,\ell'}$ is a (dimensionless) complex number characterizing the individual couplings.

\subsubsection{Implementation of hardware mixers}

Let us first discuss the implementation of the mixing unitary $U_M(\beta_k) = {\rm exp}[-i \beta_k H_M]$, cf. Eqs.~(\ref{eq:UM}) and (\ref{eq:qudit_mixer}). Specifically, we will consider two specific cases, namely $r=1$ and $r=2$, and we further comment on $r>2$.\\
~\\

\noindent
$\pmb{r=1.}$ We propose to implement the $r=1$ case as shown in Fig.~\ref{fig:level_scheme}c. 
Starting with all qudit levels in the $\mathbb S$-manifold, we apply the following sequence of pulses:
\begin{enumerate}
	\item Apply $\left \lfloor \frac{d}{2} \right \rfloor$ $\pi$-polarized $\pi$-pulses on the levels $\{1,3,\ldots, l \}$, $l = 2 \left \lfloor \frac{d}{2} \right \rfloor - 1$, on the ${\mathbb S}-{\mathbb P}$ transition, which brings them from $\mathbb{S}$ to ${\mathbb P}$.
	\item Apply the Rabi frequencies $\Omega_{\ell,\ell+1}$ connecting nearest-neighbour qudit levels for time $\tau_k$.
	\item Repeat the step 1. to bring the $\left \lfloor \frac{d}{2} \right \rfloor$ levels from ${\mathbb P}$ back to ${\mathbb S}$.
\end{enumerate}
In Fig.~\ref{fig:level_scheme}c, the choice of couplings is depicted by the red arrows and the coupling Hamiltonian reads
\begin{equation}
	\Omega h_\Omega = 
	\Omega
	\begin{pmatrix}
		0 & c_{0,1} & & \\
		c^*_{0,1} & & \ddots \\
		& \ddots & & c_{d-2,d-1} \\
		& & c^*_{d-2,d-1} & 0
	\end{pmatrix}.
	\label{eq:H_Omega}
\end{equation}
When the individual Rabi frequencies $\Omega_{\ell,\ell'}$ are adjusted such that $c_{\ell, \ell'=\ell+1} = 1 \; \forall \, \ell$, $h_\Omega \rightarrow h_M$ and the sequence 1.-3. corresponds to the mixing unitary $U_M$, Eq.~(\ref{eq:UM}), with the mixing parameter $\beta_k = \Omega \tau_k$~\footnote{For large $d$ (i.e. close to $d_{\rm max}=10$), achieving $c_{\ell, \ell'=\ell+1} = 1, \; \forall \ell$,  might be challenging. This is because the ratio $\Omega_{0,1}/\Omega_{9,10}$ between the Rabi frequencies coupling the smallest and largest $\ell$ states is too large (or too small, depending on the orientation of the magnetic field) so that $c_{\ell, \ell'=\ell+1} = 1, \; \forall \ell$, implies either slow operation or further detrimental decoherence due to the off-resonant scattering of the intense mixer beams \cite{Urech1}. In this case the mixers could be implemented in a step-wise fashion by means of the Givens rotations \cite{OLeary_2006_PRA}
}. 

One important remark is in order: the mixing Hamiltonian (\ref{eq:mixer_matrix}) implements a ``periodic boundary condition'' in that it couples the levels $0$ and $d-1$ (and similarly for higher $r$). Such coupling is typically not native to a physical qudit, which instead corresponds to ``open boundary condition'' as is apparent from (\ref{eq:H_Omega}). We discuss the difference between using (\ref{eq:mixer_matrix}) or (\ref{eq:H_Omega}) below, cf. Sec.~\ref{sec:mixer_comparison} and Fig.~\ref{fig:comp_mixers}.\\

~\\
$\pmb{r=2.}$ For $r=2$ we consider an \emph{off-resonant} coupling exploiting a detuning $\Delta$ from the ${\mathbb P}$-states as shown in Fig.~\ref{fig:level_scheme}c. This allows to extend the structure of the coupling Hamiltonian Eq.~(\ref{eq:H_Omega}) to include also a second diagonal. The coupling Hamiltonian (\ref{eq:H_Omega}) becomes
\begin{equation}
	\tilde{\Omega} h_\Omega =
	\tilde{\Omega}
	\begin{pmatrix}
		0 & \tilde{c}_{0,1} & \tilde{c}_{0,2} & & \\
		\tilde{c}^*_{0,1} & & \ddots & \ddots & \\
		\tilde{c}^*_{0,2} & \ddots & & \ddots & \tilde{c}_{d-3,d-1} \\
		 & \ddots & \ddots & & \tilde{c}_{d-2,d-1} \\
		& & \tilde{c}^*_{d-3,d-1} & \tilde{c}^*_{d-2,d-1} & 0
	\end{pmatrix}.
	\label{eq:H_Omega2}
\end{equation}
Here the tilde denotes the effective Rabi frequency which in the far-detuned limit is given by $\tilde{\Omega}_{\ell,\ell'} = \Omega_{\ell,\bar{\ell}} \Omega_{\bar{\ell},\ell'}/{\Delta_{\bar{\ell}} }$, where $\ket{\bar{\ell}}$ is the detuned state from the ${\mathbb P}$-manifold to which $\ket{\ell}$ and $\ket{\ell'}$ are coupled. As for $r=1$, here the Rabi frequencies $\tilde{\Omega}_{\ell,\ell'}$ have to be adjusted such that $\tilde{c}_{\ell,\ell+1} = \tilde{c}_{\ell,\ell+2}=1 \; \forall \ell$.\\

~\\
$\pmb{r > 2.}$ Going beyond $r=2$ becomes non-trivial due to the connectivity of the coupling Hamiltonian (\ref{eq:H_Omega2}), here limited to next-to-nearest qudit levels. A general strategy, exploiting the decomposition of an arbitrary single qudit unitary into a sequence of parallelized Givens rotations under the finite connectivity constraint and invoking a greedy optimization algorithm has been described in~\cite{OLeary_2006_PRA}. As in this work we are concerned only with $r \leq 2$ we do not analyse this situation further.

\subsubsection{Performance of hardware mixers}
\label{sec:mixer_comparison}
To investigate how the hardware-specific mixers given by Eqs. \eqref{eq:H_Omega} and \eqref{eq:H_Omega2} compare to the mixer \eqref{eq:mixer_matrix} we run simulations similar to the ones performed in Sec.~\ref{sec:P}.
To this end we consider the $N=4$ complete graph data set. This choice is motivated by the fact that for $N=3$ (i.e. $d=3$), Eq. \eqref{eq:H_Omega2} becomes \eqref{eq:mixer_matrix} and as $N$ increases (and $d=N$ for complete graphs considered here), the two boundary states $\ket{0},\ket{d-1}$ that have different mixing constitute an increasingly small fraction over all states. It is thus plausible to assume that the difference between the mixers Eqs.~\eqref{eq:mixer_matrix},\eqref{eq:H_Omega} and \eqref{eq:H_Omega2} is most relevant for smallest $N>3$, i.e. $N=4$.

All non-zero elements in \eqref{eq:H_Omega} and \eqref{eq:H_Omega2} are set to unity, and we again generate initial points starting from the all-negative-weights graph. The numerical results in Fig.~\ref{fig:comp_mixers} show that the performance is very similar amongst the different mixers -- the largest observed percentage performance difference for a single instance is about $4\%$. From the data we can conclude that the hardware mixer with $r=2$ performs better on average than the normal mixer with $r=1$, which outperforms the hardware mixer with $r=1$. Overall, we can conclude from this data that replacing the standard mixer \eqref{eq:mixer_matrix} with the hardware-specific mixers \eqref{eq:H_Omega} and \eqref{eq:H_Omega2} is not expected to result in a significant decrease in performance.
\begin{figure}
\centering
    \includegraphics[width=\linewidth]{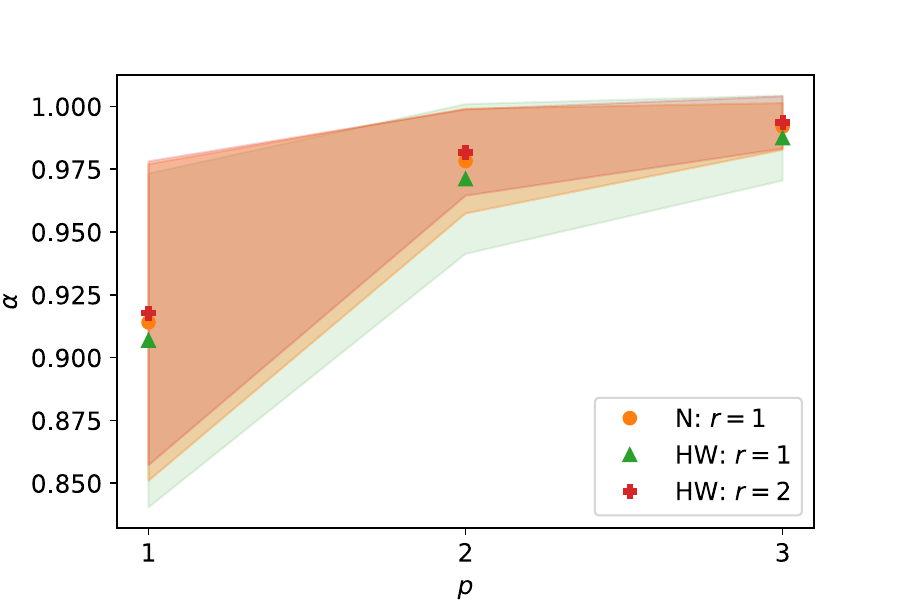}
    \caption{Results for the different mixing Hamiltonians given by Equations \eqref{eq:mixer_matrix},\eqref{eq:H_Omega} and \eqref{eq:H_Omega2}, labelled as `N: $r=1$', `HW: $r=1$' and `HW: $r=2$', respectively. A dot is the average approximation ratio over all 50 instances and the shaded area represents one standard deviation. The results are obtained using a complete graph with $N=4$ data set in the same way as described in Sec.~\ref{sec:P}. }
    \label{fig:comp_mixers}
\end{figure}

\subsubsection{Unitaries beyond mixers}
Next, we specify single qudit unitaries beyond mixers, which we will exploit in the construction of the 2-gates in Sec.~\ref{sec:RQC_2gates} below.
For a two-level system with levels $\ket{\ell},\ket{\ell'}$ driven on resonance with Rabi frequency $\Omega_{\ell,\ell'}$, the unitary evolution operator expressed in the $\{\ket{\ell}, \ket{\ell'} \}$ basis reads (cf. Appendix~\ref{app:RQC_BuildingBlocks})
\beq
    U^{\rm 2-level}_{\ell, \ell'}(\theta,\varphi) = 
    \begin{pmatrix}
     \cos \frac{\theta}{2} & -i {\rm e}^{i \varphi} \sin \frac{\theta}{2} \\
     -i {\rm e}^{-i \varphi} \sin \frac{\theta}{2} & \cos \frac{\theta}{2}
    \end{pmatrix},
    \label{eq:U2level}
\eeq
where $\theta = 2 t |\Omega_{\ell,\ell'}|$, $t$ is the duration of the laser pulse and $\Omega_{\ell,\ell'} = |\Omega_{\ell,\ell'}| {\rm e}^{i \varphi}$.

Similarly, we can write down a general unitary evolution operator for a three-level system (so-called $\Lambda$-scheme). As we will be particularly interested in performing controlled operations through driving the qudit states to a Rydberg level, we shall consider a $\Lambda$-scheme where levels $\ket{\ell}, \ket{\ell'}$ are coupled to a common Rydberg state $\ket{r}$. 
We will further require that at the end of the operation, there is no population in the Rydberg state. 
In this case the unitary for \emph{resonant} driving and expressed
in the $\{\ket{\ell}, \ket{\ell'}, \ket{r} \}$ basis reads (cf. Appendix~\ref{app:RQC_BuildingBlocks})
\beq
    U^{\rm 3-level}_{\ell,\ell'}(\theta,\varphi) =  
    -\begin{pmatrix}
     \cos \frac{\theta}{2} & {\rm e}^{i \varphi} \sin \frac{\theta}{2} & 0\\
     {\rm e}^{-i \varphi} \sin \frac{\theta}{2} & -\cos \frac{\theta}{2} & 0 \\
     0 & 0 & 1
    \end{pmatrix},
    \label{eq:U3level}
\eeq
where 
\begin{subequations}
    \label{eq:theta_def}
    \begin{align}
	\cos \frac{\theta}{2} &= \frac{|\Omega_0|^2 - |\Omega_1|^2}{\Omega^2}, \\
	\sin \frac{\theta}{2} \, {\rm e}^{i \varphi} &= \frac{2 \Omega_0 \Omega_1^*}{\Omega^2} \\
	\Omega &= \sqrt{|\Omega_0|^2 + |\Omega_1|^2}
	\end{align}
\end{subequations}
and $\Omega_\alpha = |\Omega_\alpha| {\rm e}^{i \phi_\alpha}$, $\alpha=0,1$ such that $\varphi = \phi_0 - \phi_1$. Here, we have denoted $\Omega_0 = \Omega_{\ell,r}, \Omega_1 = \Omega_{\ell',r}$ for simplicity. The unitary takes the form (\ref{eq:U3level}) for pulses of duration $t = \pi (2 m-1)/\Omega$, $m \in {\mathbb N}$.
In what follows we refer to the unitaries (\ref{eq:U2level}) and (\ref{eq:U3level}) simply as $U$ and we shall specify which one is considered where appropriate. 
%
For future convenience we denote the usual Pauli X as 
\beq
    X_{\ell,\ell'} = U_{\ell, \ell'}(\pi,0)
    \label{eq:PauliX}
\eeq
and the phase gate (defined up to a global phase)
\beq
    {\rm P}_{\ell,\ell'}(\varphi) = U_{\ell,\ell'}(\pi,\varphi_2) U_{\ell,\ell'}(\pi,\varphi_1) 
    =
    \begin{pmatrix}
        {\rm e}^{i \varphi} & 0 \\
        0 & {\rm e}^{-i \varphi}
    \end{pmatrix},
    \label{eq:P}
\eeq
where $\varphi = \varphi_2-\varphi_1$.

\subsection{2-gates}
\label{sec:RQC_2gates}

\noindent
\textit{CP gate.} We proceed with the construction of the cost unitary $U_C$, Eq.~(\ref{eq:UC}). Noting that the the cost Hamiltonian (\ref{eq:HC}) is given by a sum of commuting operators acting on the graph edges, we consider an action of the cost unitary on a single pair of qudits. It corresponds to a controlled-phase gate of the form
\beq
    {\rm CP}(\gamma) \ket{\ell}\ket{\bar{\ell}} = \left[ \delta_{\ell,\bar{\ell}} {\rm e}^{-i \gamma} + (1-\delta_{\ell,\bar{\ell}}) \right] \ket{\ell}\ket{\bar{\ell}}.
    \label{eq:CP_gate}
\eeq
It is defined up to a global phase and $\delta_{\ell,\bar{\ell}}$ is the Kronecker delta
\footnote{
We note that according  to the definitions Eqs. (\ref{eq:UC}), (\ref{eq:defV}) and (\ref{eq:HC}), the global phase omitted in the Eq.~(\ref{eq:CP_gate}) equals $\mp \gamma_k$ for the edge weight $w_{(u,v)}=\pm 1$ such that the relative phase $\gamma$ between the two-qudit states $\ket{\ell} \ket{\bar{\ell}=\ell}$ and $\ket{\ell} \ket{\bar{\ell} \neq \ell}$ is given by $\gamma = \pm 2 \gamma_k$. We also remark that while we consider the graph edges weights $w_{(u,v)} \in \{-1,1\}$, one could implement arbitrary weights through edge-specific phases $\gamma \rightarrow \gamma_{(u,v)} = 2 w_{(u,v)} \gamma_k$.
}. 

The key element in the construction of the CP gate Eq.~(\ref{eq:CP_gate}) is a controlled-phase unitary ${\cal U}$ acting on a single level $\ket{\ell}$ from the qudit manifold by coupling it to the corresponding Rydberg state $\ket{r_\ell}$, namely
\beq
\label{eq:calU_def}
\centering
	{\cal U}^{(1|2)}_{\ell}(\gamma) \equiv
	\parbox{0.4\textwidth}{
\Qcircuit @C=1em @R=1em
{
& {\rm t}\phantom{ar} & \qw & \gate{{\rm P}_{\ell,r_\ell}(\frac{\gamma}{2})} & \push{\rule{0em}{1.5em}} \qw & \qw \\
& {\rm c}\phantom{on} & \gate{{\rm X}_{\ell,r_\ell}} & \qw & \gate{{\rm X}_{\ell,r_\ell}} & \qw
\gategroup{1}{3}{2}{5}{1.5em}{--}
}
},
\eeq
where the first and second qudit is the target and control respectively (labelled as t and c) and P is given in Eq.~(\ref{eq:P}). 
Here, we have introduced the notation ${\cal U}^{(q_t|q_c)}$, where $q_{t}$, $q_c$ label the target and control qudits, respectively. We can then construct the controlled phase gate Eq.~(\ref{eq:CP_gate}) either in a manifestly symmetric way
\beq
	{\rm CP}(\gamma) = \prod_{\ell=0}^{d-1} {\cal U}^{(1|2)}_\ell {\cal U}^{(2|1)}_\ell,
\label{eq:CP_symmetric}
\eeq
or alternatively as
\beq
    {\rm CP}(\gamma) = \left[ \prod_{\ell=0}^{d-2} {\rm CX}^{(1|2)}_{\ell, \ell+1 | \neg \ell} \right]^\dag P^{(1)}_{{\rm aux},0} (\gamma) \left[ \prod_{\ell=0}^{d-2} {\rm CX}^{(1|2)}_{\ell, \ell+1 | \neg \ell} \right],
    \label{eq:CP_Kj}
\eeq
cf. Appendix~\ref{app:RQC_BuildingBlocks} for derivation.
The gate (\ref{eq:CP_Kj}) has the advantage of reducing the cost of the gate compared to (\ref{eq:CP_symmetric}) in terms of the number of hardware operations (laser pulses), cf. Eq.~(\ref{eq:CP_cost}) and Sec.~\ref{sec:PD}.
Here, $P^{(1)}_{{\rm aux},0}$ is the phase gate (\ref{eq:P}) applied to the first (target) qudit and driving the level $\ket{0}$ through an auxiliary state, which here is \emph{not} a Rydberg state. Similarly, we have introduced the controlled-X gates ${\rm CX}^{(q_t|q_c)}_{\ell_t, \ell_t' | \neg \ell_c}$ where $q_c$ denotes the control, $q_t$ the target, $\ell_c$ the control qudit level and $\ell_t, \ell_t'$ the pair of target levels being acted upon.
Importantly, here the target levels $\ell_t, \ell_t'$ are swapped when the control qubit is \emph{not} in the state $\ell_c$ which is highlighted using the negation sign in the subscript, $\neg \ell_c$.
We note that for a typical hardware implementation, the choice of $\ell_t, \ell_t'$ is not arbitrary but upper bounded. In the present case we shall assume $|\ell_t - \ell_t'| \leq 2$. The controlled-X are implemented in the standard fashion as
\beq
\label{eq:CX_def}
\centering
	{\rm CX}^{(1|2)}_{\ell_t, \ell_t' | \neg \ell_c} \equiv
	\parbox{0.3\textwidth}{
\Qcircuit @C=1em @R=1em
{
& {\rm t}\phantom{ar} & \qw & \gate{{\rm X}_{\ell_t,\ell_t'}} & \push{\rule{0em}{1em}} \qw & \qw \\
& {\rm c}\phantom{on} & \gate{{\rm X}_{\ell_c,r_{\ell_c}}} & \qw & \gate{{\rm X}_{\ell_c,r_{\ell_c}}} & \qw
\gategroup{1}{3}{2}{5}{1.5em}{--}
}
},
\eeq
where the X-gate acting on the target is driven through a Rydberg level (a $\Lambda$-scheme) as described by the Eq.~(\ref{eq:U3level}).
We also remark that the CP gate Eq.~(\ref{eq:CP_Kj}) is invariant under the exchange of the control and target qudits, i.e. $(1|2) \rightarrow (2|1)$.

In Sec.~\ref{sec:PD} we will be concerned with assessing the cost of the algorithm as counted by the number of elementary 2-gate operations. To this end we introduce the notation $[{\rm G}]$ to denote the cost of a gate ${\rm G}$ as counted by the number of the ${\rm CX}$ gates (\ref{eq:CX_def}) or its equivalents in the decomposition of ${\rm G}$. The corresponding cost of the CP gates is thus
\footnote{For a qubit, the cost of the symmetric gate Eq.~(\ref{eq:CP_symmetric}) can be further reduced to $[{\rm CP}]=d$, cf. Appendix~\ref{app:RQC_BuildingBlocks}}
\begin{subequations}
    \begin{align}
        {\rm Eq. \; (\ref{eq:CP_symmetric}):} \; [{\rm CP}] &= 2 d \\
        {\rm Eq. \; (\ref{eq:CP_Kj}):} \; [{\rm CP}] &= 2 d-2 \label{eq:cost_CP_Kj}
    \end{align}
    \label{eq:CP_cost}
\end{subequations}


~\\
\textit{SWAP gate.} In order to perform the CP gate on a pair of distant graph vertices, it is in general necessary to bring them together by means of a swap gate ${\rm SWAP}_d$. Several possible implementations of the SWAP gate for qudits have been proposed~\cite{Garcia_2013_QIP,Balakrishnan_2014}. Here we shall consider the implementation of Ref.~\cite{Garcia_2013_QIP}, which parallels the qubit SWAP construction consisting of three consecutive CX gates: ${\rm SWAP}_2 = {\rm CX}^{(q|\bar{q})} {\rm CX}^{(\bar{q}|q)} {\rm CX}^{(q|\bar{q})}$, where we have dropped the level indices taking $\ket{1}$ to be the control level for the qubit as customary. The qudit version of the SWAP is given by~\cite{Garcia_2013_QIP}
\beq
    {\rm SWAP}_d = {\rm CX}_d^{(q|\bar{q})} {\rm CX}_d^{(\bar{q}|q)} {\rm CX}_d^{(q|\bar{q})},
    \label{eq:SWAP_d}
\eeq
which is the direct generalization of the SWAP for qubits. The qudit controlled-X ${\rm CX}_d$ is defined as
\beq
    {\rm CX}_d \ket{\ell}\ket{\bar{\ell}} = \ket{\ell}\ket{-\ell-\bar{\ell}}
\eeq
and mod $d$ is understood in evaluating $-\ell-\bar{\ell}$ in the last expression. It can be implemented as
\beq
    {\rm CX}_d^{(q | \bar{q})} = {\rm QFT}_d^{(q)} {\rm CZ}_d^{(q | \bar{q})} {\rm QFT}_d^{(q)},
\eeq
where
\beq
    {\rm QFT}_d \ket{\ell} = \frac{1}{\sqrt{d}} \sum_{\ell'=0}^{d-1} {\rm e}^{i \frac{2\pi}{d} \ell \ell'} \ket{\ell'}
    \label{eq:QFT}
\eeq
is the quantum Fourier transform in the single qudit space (i.e.\ a unitary operation with matrix elements $ \left( {\rm QFT}_d \right)_{\ell', \ell} = {\rm exp}[i (2\pi/d) \ell \ell']/\sqrt{d} $) and ${\rm CZ}_d$ is defined as
\beq
    {\rm CZ}_d \ket{\ell} \ket{\bar{\ell}} = {\rm e}^{i \frac{2\pi}{d} \ell \bar{\ell}} \ket{\ell} \ket{\bar{\ell}}.
    \label{eq:CZd_def}
\eeq
We note that the ${\rm QFT}_d$ gate is a 1-gate and can thus be synthesized using methods of~\cite{OLeary_2006_PRA}, similarly to the mixing unitary for $r>2$, cf. also~\cite{Stroud_2002_JModOpt} for implementation of ${\rm QFT}_d$ in the context of multilevel atoms. To quantify the cost of the SWAP gate, we thus focus on the ${\rm CZ}_d$ gate (\ref{eq:CZd_def}), which can be implemented as
\beq
    {\rm CZ}_d^{(q|\bar{q})} = \prod_{\bar{\ell}=1}^{d-1} X^{(\bar{q})}_{\bar{\ell},r_{\bar{\ell}}} \left[ \prod_{\ell=1}^{d-1} {\rm P}^{(q)}_{\ell,r_\ell}\left(\frac{2\pi}{d} \ell \bar{\ell} \right) \right] X^{(\bar{q})}_{\bar{\ell},r_{\bar{\ell}}},
\eeq
which contains $(d-1)^2$ applications of the phase gate P. Importantly, we note that the product in the brackets can be executed \emph{in parallel} by simultaneous application of the laser pulses connecting each level $\ket{\ell},\; \ell = 1,\ldots,d-1,$ of the target to its respective Rydberg level $\ket{r_\ell}$. This is precisely an example of the parallelization offered by the qudit hardware. We thus get, in conjunction with Eq.~(\ref{eq:SWAP_d}), for the total cost of the qudit SWAP
\beq
    [{\rm SWAP}_d] = 3(d-1).
    \label{eq:cost_SWAPd}
\eeq


\subsection{Initialisation and measurement}

\noindent
\emph{Initialization.} The initial state (\ref{eq:ML_is}) can be prepared by initializing all atoms in the $\ket{0} \equiv \ket{{\mathbb S},m_F=-9/2}$ state by standard means of optical pumping and then applying a sequence of $d-1$ unitaries (\ref{eq:U2level}) such that
\beq
    \ket{+^d} = \left[ \prod_{\ell=1}^{d-1} U_{\ell-1,\ell}^\dag \left( \theta_\ell, \varphi=\frac{\pi}{2} \right) \right]^\dag \ket{0},
\eeq
where $\cos (\theta_\ell/2) = 1/\sqrt{d-(\ell-1)}$. Here the unitary can be implemented either via resonant or off-resonant Raman coupling described in Sec.~\ref{sec:RQC_1gates} for the range $r=1$ and $r=2$ of the mixing unitaries respectively.

~\\
\noindent
\emph{Measurement.}
In order to projectively measure the quantum state after each iteration of the QAOA we consider imaging the atoms using resonance fluorescence by collecting the light scattered from the strongly driven $\ket{0}-\ket{^3{\rm P}_1, F'=11/2, m_{F'}=-11/2}$ transition, cf. Fig~\ref{fig:level_scheme}a.
This has the advantage of simultaneously cooling the atoms in the $\ket{0}$ state while imaging them~\cite{Urech2}, cf. also~\cite{Saskin_2019_PRL, Covey_2019_PRL,madjarov2021entangling} for related techniques.
Exploiting the state-specific detection of individual $m_F$ states, first only the population in the $\ket{0}$ state ($m_F=-9/2$) is being detected. In the case of negative outcome (no population in the $\ket{0}$ state), the population from the adjacent $m_F$ state, i.e.\ from $\ket{1}$ ($m_F =- 7/2$), is transferred to $\ket{0}$ by using optical pumping, or coherently via stimulated Raman adiabatic passage and the $\ket{0}$ is imaged again. This process is repeated until the positive detection of some qudit level $\ket{\ell}$.
This allows for imaging of all the qudit states within the expected lifetime in optical tweezers of $\gtrsim 10$ seconds (as there is no active cooling of the $\ket{1}, \ldots, \ket{d-1}$ states during the imaging). This time is limited mainly by off-resonant scattering of the tweezer light and also by background gas collisions, where the latter can be further reduced by increasing the quality of the vacuum.


\section{\label{sec:PD} Processor design and gate count}

Here we seek to evaluate the cost of the algorithm, cf. Sec.~\ref{sec:Recap} and Sec.~\ref{sec:Int}, as quantified by the gate count. As discussed in Appendix~\ref{app:Errors}, the dominant errors stem mainly from the 2-gates and we thus focus on the 2-gate count. To this end we consider the total count ${\cal C}_{\rm tot}$ of \emph{qubit} primitive 2-gates and specifically we will use the qubit controlled-X $[{\rm CX}]$ as our cost unit (This choice of counting will be useful when comparing the qudit vs. qubit encodings in Sec.~\ref{sec:PD_Comparison}).
${\cal C}_{\rm tot}$ is determined by \emph{(i)} the topology of the graph encoding the clustering problem, \emph{(ii)} the topology of the quantum processor and, for qubits, \emph{(iii)} the encoding scheme, which we discuss in Sec.~\ref{sec:PD_Comparison}.

In this section we will find that taking into account the hardware considerations \emph{(i)-(iii)} yields the expected result, namely that the qudit encoding is superior to the qubit one in that it yields smaller ${\cal C}_{\rm tot}$. Readers not interested in the details of the gate count can consult the summary in the Table~\ref{tab:GateCount}.

To proceed, let us first comment on 2-gates beyond nearest neighbours. The neutral atom and ion based systems possess a native \emph{long-range} interaction, which typically decays as a power law $1/R^\alpha$ with distance ($\alpha=6$ for a pair of Rydberg atoms interacting through a Van der Waals potential). Such potential gives rise to clusters of higher connectivity on the processor, which can lead to an improvement in performance over processors with only nearest-neighbour interactions~\cite{linke2017experimental}. 
The related scaling properties of the quantum gates for trapped ensembles rather than single neutral atoms have been analysed theoretically in~\cite{Saffman_2008_PRA} and such ensembles occupying hundreds of sites have been realized recently experimentally~\cite{wang2020preparation}. Furthermore, claims have been made that up to 50 atoms can be connected without the need of a SWAP operation~\cite{ColdQuanta}. In the here considered implementation using Rydberg blockade and Van der Waals interaction, the price to pay for the higher connectivity is the longer duration of the gate. Taking the basic building block Eq.~(\ref{eq:calU_def}) of the CP gate as an example, the gate duration is mainly given by the duration of the P-gate applied to the target qudit. This is because while the X-gates acting on the control qudit can be performed in principle arbitrarily fast limited only by the available Rabi frequency, the Rabi frequency $\Omega$ used to realize the P-gate has to satisfy the blockade constraint $V(R)/\Omega \gg 1$ for a given atom distance $R$. Since $V(R) \propto 1/R^6$, for the same quality of blockade (same $V(R)/\Omega$) the gate time thus scales as $\propto 1/\Omega \propto R^6$. For this reason and for the sake of clarity, in the following we consider only nearest-neighbour 2-gates and leave the algorithm hardware optimization exploiting longer-range connectivity for future work.

Another remark is that it is desirable to \emph{parallelize} the gate operations, cf. e.g. \cite{Levine_2019_PRL} for recent realization with Rydberg atoms, to reduce the absolute time of the algorithm run. This in principle allows one to reduce the effect of noise such as the background gas collisions or off-resonant scattering from the optical traps, cf. Sec.~\ref{sec:Errors}. Such parallelization however does not change the 2-gate count and we don't elaborate on it further.


\makeatletter\onecolumngrid@push\makeatother

\begin{table}
 \centering
   \begin{adjustbox}{width=\textwidth}
\begin{tabular}{|c|c|c|c||c|}
\hline
& \multicolumn{3}{c||}{Processor geometry} & Gates for binary encoding \\ \hline
\multirowcell{2}[-1cm]{\rotatebox{90}{1D}} & \parbox[c]{2mm}{\rotatebox[origin=c]{90}{qudit}} & 
	\multicolumn{2}{c||}{ \parbox[c][2cm]{5cm}{\includegraphics[scale=0.4]{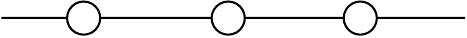}} } & \multirow{2}*{\parbox[c][4cm]{9cm}{
\beq
\underset{{\frak q}=\log_2 d}{\widetilde{\text{CP}}} = 
\raisebox{1cm}{
\scalebox{0.43}{
\Qcircuit @C=0.97em @R=1.2em{
&&& 1  & &  \qw & \ctrl{6} & \qw & \dots &  & \qw & \qw & \dots &   & \qw &\qw &\qw & \qw &\qw &\qw & \qw &\qw & \dots &  & \qw & \qw & \dots &  &\ctrl{6}& \qw\\
&&& \vdots & &&   &        & \ddots & &     &     &  &   &   &     &    &     &     &       &     &    &      &   &     & & \iddots &  & &\\   
&&& {} & & \qw &\qw & \qw & \dots &  & \ctrl{6} & \qw & \dots & & \qw &\qw&\qw & \qw &\qw &\qw & \qw &\qw & \dots &  & \ctrl{6} & \qw & \dots &  &\qw& \qw\\
&&&\vdots & & & &     &       &   &        &  &\ddots&   &  &     &     &   &    &    &     &     & \iddots &   & &        &        & & &\\
&&& {\frak q} & & \qw &\qw & \qw & \dots &  & \qw & \qw & \dots  & & \ctrl{6} &\qw & \qw &\qw&\qw & \qw &\ctrl{6} &\qw  & \dots &  & \qw & \qw & \dots &  &\qw& \qw\\
{}  & &  & & & &   & &  &       &   &          &&&   &  &     &     &   &    &     &   &     &     &  &   &        &        & & & &\\
&&& 1 & & \qw & \targ \qw & \qw &  \dots &  & \qw & \qw & \dots &   & \qw &\gate{X} & \qw &\ctrl{2} & \qw & \gate{X} & \qw &\qw & \dots &  & \qw & \qw & \dots &  &\targ \qw & \qw\\
&&&\vdots & & & &         & \ddots & &     &    &     &   &     & \vdots   &     &    &    &  \vdots   &     &    &      &   &    &  & \iddots &  & & \\   
&&& {} & &  \qw & \qw & \qw  & \dots &  & \targ \qw & \qw & \dots &  & \qw &\gate{X} & \qw &\ctrl{2} & \qw & \gate{X} & \qw &\qw & \dots &  & \targ \qw &\qw &  \dots &  &\qw & \qw\\
&&&\vdots & & &  &    &       &           & &  &\ddots&    &     &  \vdots   &   &    &     &\vdots   &     &     & \iddots &   &        &        & & & \\
&&& {\frak q} & & \qw &\qw & \qw  & \dots &  & \qw         & \qw & \dots &  & \targ \qw &\gate{X} & \qw &\gate{U_3(0,-\phi,0)} & \qw & \gate{X} &\targ \qw &\qw  & \dots & & \qw & \qw & \dots &  &\qw & \qw \gategroup{1}{6}{5}{30}{1.9em}{-} \gategroup{7}{6}{11}{30}{1.9em}{-}\\
}
}
}
\nonumber
\eeq
}
} \\ \cline{2-4}
& \parbox[c]{2mm}{\rotatebox[origin=c]{90}{binary}} & 
	\multicolumn{2}{c||}{\parbox[c][2cm]{5cm}{\includegraphics[scale=0.5]{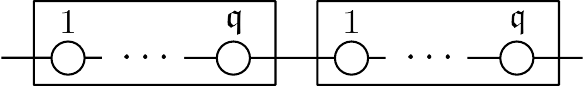}}} & \\ \cline{1-4} \cline{1-4}
\multirowcell{2}[-6cm]{\rotatebox{90}{2D}} & \parbox[c]{2mm}{\rotatebox[origin=c]{90}{qudit}} & \multicolumn{2}{c||} {\parbox[c][4cm]{5cm}{\includegraphics[scale=0.5]{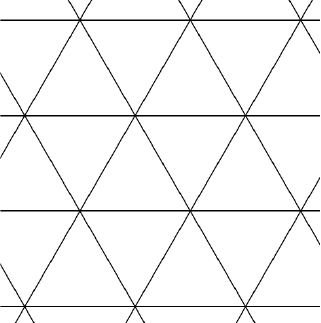}}} & \multirow{2}*{\parbox[c][5cm]{7cm}{
\beq
\underset{{\frak q} \neq \log_2 d}{\widetilde{{\rm CP}}} =
\raisebox{2cm}{
\scalebox{0.6}{
\Qcircuit @C=0.9em @R=0.9em {
& & & 1 & & \multigate{4}{\ N_1\ } & \ctrl{12} & \qw & \qw &\qw & \qw & \ctrl{12} & \multigate{4}{\ N_1\ } & \qw\\
& & &\vdots& &  & & & & & & & & \\
& & & {} & & \ghost{\ N_1\ } \qw & \ctrl{10} & \qw & \qw &\qw & \qw & \ctrl{10} & \ghost{\ N_2\ } \qw & \qw\\
& & &\vdots& &  & & & & & & & & \\
& & & {\frak q} & & \ghost{\ N_1\ }  & \ctrl{7} & \qw & \qw  &\qw & \qw & \ctrl{7} & \ghost{\ N_2\ }  & \qw\\
{} & & & & & & & & & & & & &\\
& & & 1 &  & \multigate{4}{\ N_2\ } &\qw  & \ctrl{7} & \qw & \ctrl{7} & \qw & \qw & \multigate{4}{\ N_2\ } & \qw \\
& & & \vdots& &  & & & & & & & & \\
& & & {} &  &\ghost{\ N_2\ } \qw & \qw & \ctrl{5} & \qw & \ctrl{5} & \qw & \qw &\ghost{\ N_2\ } \qw & \qw \\
& & & \vdots& & & & & & & & & & \\
& & & {\frak q} &  &\ghost{\ N_2\ } \qw  & \qw & \ctrl{3} & \gate{U_3(0,-\phi,0)} & \ctrl{3} & \qw & \qw &\ghost{\ N_2\ } \qw & \qw \\
& & &  & & & & & & & & & &\\
& & &  a_0 & &\qw & \targ \qw  & \qw &  \qw & \qw &  \qw & \targ \qw & \qw & \qw\\
 & & & a_1 & &\qw & \qw & \targ \qw & \qw  &\targ \qw &\qw & \qw &  \qw & \qw \gategroup{1}{6}{5}{13}{1em}{-} \gategroup{7}{6}{11}{13}{1em}{-} \gategroup{13}{6}{14}{13}{1em}{-} 
    }
  }
} \nonumber
\eeq
}} \\	\cline{2-4}
& \parbox[c]{2mm}{\multirow{3}{*}{}} & \rotatebox[origin=c]{90}{$\begin{gathered} {\frak q}=2 \end{gathered}$} & \parbox[c][3cm]{6cm}{$ \begin{gathered} \includegraphics[scale=0.5]{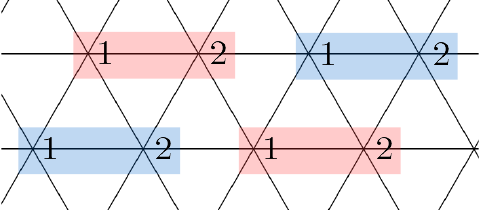} \end{gathered}$} & \\	\cline{3-4}
& \rotatebox[origin=c]{90}{binary} & \rotatebox[origin=c]{90}{$\begin{gathered} {\frak q}=3 \end{gathered}$} & \parbox[c][5cm]{7cm}{$\begin{gathered} \includegraphics[scale=0.4]{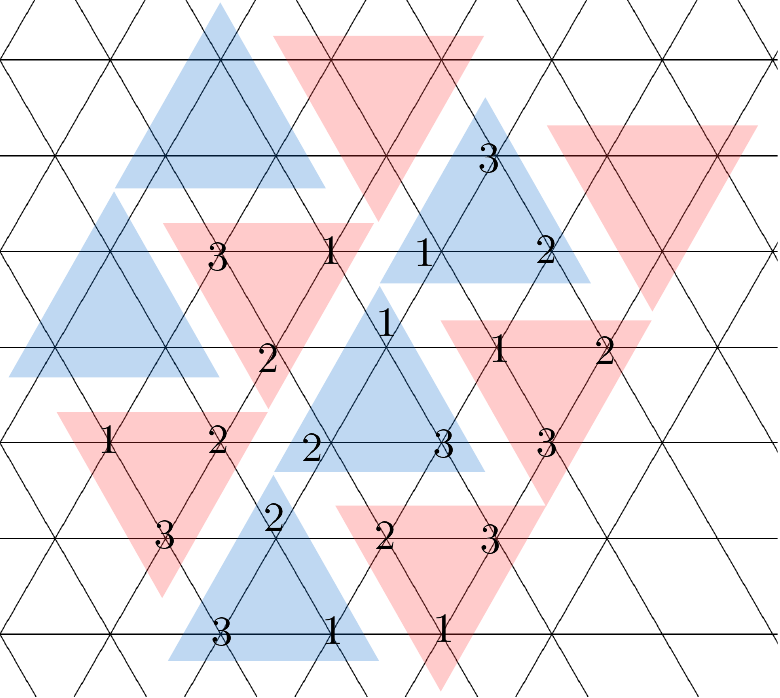} \end{gathered}$} & \parbox[c][5cm]{7cm}{
\beqa
&& \phantom{XXXXX}
\scalebox{0.85}{
\Qcircuit @C=.4em @R=0em @! {
 \lstick{1} & \ctrl{1} & \qw & \\
 \lstick{{\rm C}^2(U) = 2} & \ctrl{1} & \qw & \\
 \lstick{3} & \gate{U} & \qw & 
}
}
\nonumber \\[0.3cm]
&&
\scalebox{0.8}{
\Qcircuit @C=.4em @R=0em @! {
\lstick{1} & \qw & \ctrl{1} & \qw & \ctrl{1} & \ctrl{2} & \qw\\
\lstick{\phantom{{\rm C}^2(U)} = 2} & \ctrl{1} & \targ & \ctrl{1} & \targ & \qw & \qw\\
\lstick{3} & \gate{V} & \qw & \gate{V^\dag} & \qw & \gate{V} & \qw
}
}
\nonumber
\eeqa
} \\ \cline{3-4}
&& \rotatebox[origin=c]{90}{$\begin{gathered} {\frak q}=4 \end{gathered}$} & \parbox[c][4cm]{7cm}{$\begin{gathered} \includegraphics[scale=0.4]{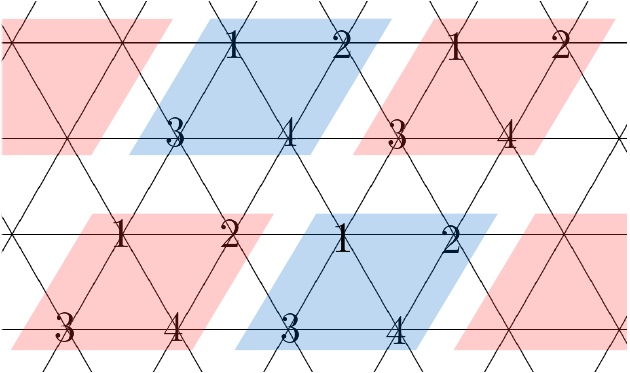} \end{gathered}$} & \hspace*{-0.5cm}\parbox[c][5cm]{7cm}{
\beqa
&& \phantom{XXXXXXXXXXX}
\scalebox{0.85}{
\Qcircuit @C=.3em @R=0em @! {
& \lstick{1}&               \ctrl{1} & \qw & \\
& \lstick{2} & \ctrl{1} & \qw & \\
\lstick{\raisebox{2.2em}{$\text{C}^3(U)=$}} & \lstick{3}&  \ctrl{1} & \qw & \\      
& \lstick{4}&                 \gate{U} & \qw &
}
}
\nonumber \\[0.5cm]
&&
\scalebox{0.65}{
\Qcircuit @C=.3em @R=0em @! {
& \lstick{1} & \ctrl{3} & \ctrl{1} & \qw & \ctrl{1} & \qw & \qw &\qw &\ctrl{2} & \qw & \qw &\qw & \ctrl{2} & \qw &\qw \\
& \lstick{2} & \qw      & \targ    & \ctrl{2} & \targ &\ctrl{2} & \ctrl{1} & \qw & \qw &\qw & \ctrl{1} &\qw &\qw &\qw &\qw\\
\lstick{\raisebox{2.2em}{$\phantom{\text{C}^3(U)}=$}} & \lstick{3} & \qw      & \qw      & \qw      & \qw   & \qw      &\targ \qw &\ctrl{1} &\targ \qw& \ctrl{1} &\targ \qw & \ctrl{1} &\targ \qw & \ctrl{1} &\qw \\
& \lstick{4} & \gate{V} & \qw   & \gate{V^\dag} & \qw & \gate{V} & \qw      & \gate{V^\dag} & \qw & \gate{V}& \qw      & \gate{V^\dag} & \qw & \gate{V} &\qw
}
}
\nonumber
\eeqa
}\\
\hline
\end{tabular}
  \end{adjustbox}
\caption{
\emph{Left column:} Considered processor geometries in 1D (upper two rows) and 2D (lower four rows) for both the qudits and the binary encoding. In 2D the shaded blue and red areas highlight the effective qudit (e-dit) in the binary encoding including the considered enumeration of the qubits.
\emph{Right column:} The construction of the $\widetilde{\rm CP}$ gates for binary encoding. Construction for both ${\frak q} = \log_2 d$ and ${\frak q} \neq \log_2 d$ case is shown (based on Ref.~\cite{fuchs2021efficient}) together with the decomposition of the multi-controlled ${\rm C}^{{\frak q}-1}(U)$ gate for ${\frak q=4}$ and ${\frak q}=3$ (based on Ref.~\cite{Barenco_1995_PRA}). The unitaries $V$ satisfy $V^2 = U$ and $V^4 = U$ in the decomposition of ${\rm C}^2(U)$ and ${\rm C}^3(U)$ respectively \cite{Barenco_1995_PRA}.
}
\label{tab:architectures}
\end{table}

\clearpage
\makeatletter\onecolumngrid@pop\makeatother

\begin{widetext}

\begin{table}
 \centering
 \begin{adjustbox}{width=\textwidth}
 \begin{tabular}{|c|c|c|c|c|c|c|c|}
 \hline
 \rotatebox[origin=c]{90}{Geometry} & \rotatebox[origin=c]{90}{Encoding} && $n^{\rm inter}_{\rm SWAP}$
 & 
	\begin{minipage}{3cm}
		\beqa
			&& [\widetilde{\rm SWAP}] = \nonumber \\
			&=& {\frak q}^2 [{\rm SWAP}_2] \nonumber \\
			&=& 3 \times {\frak q}^2 [{\rm CX}] \nonumber
		\eeqa
	\end{minipage}
 & 
	\begin{minipage}{3.5cm}
		\beqa
			&& n^{\rm intra}_{{\rm SWAP}_2} [{\rm SWAP}_2] = \nonumber \\
			&=& 3 {\frak q}({\frak q}-1) [{\rm SWAP}_2] \nonumber \\
				&=& 3 \times 3 {\frak q}({\frak q}-1) [{\rm CX}] \nonumber
		\eeqa
	\end{minipage}  
 &
 $\left[ {\rm C}^{{\frak q}-1}(U) \right]$ & ${\cal C}_{\rm tot}$ \\ 
 \hline
 \multirow{7}{*}{1D} & \parbox[t]{2mm}{\multirow{3}{*}{\rotatebox[origin=c]{90}{binary}}} 
 & \cbOne ${\frak q}=2$ & \multirow{7}{*}{$\frac{(N-1)(N-2)}{2}$} & \cbOne $3 \times 4$ & \cbOne $3 \times 6$ & \cbOne $1 [\rm CX] = 1$ & \cbOne {\color{RoyalPurple} $|E| \times  35 - O(N) \times 12$} \\ 
 \cline{3-3} \cline{5-8}
 && \cbTwo ${\frak q}=3$ && \cbTwo $3 \times 9$ & \cbTwo $3 \times 18$ & \cbTwo $5[{\rm CX}] + 2[{\rm SWAP}_2] = 11$ & \cbTwo {\color{RoyalPurple} $|E| \times 98 - O(N) \times 27$} \\ \cline{3-3} \cline{5-8}  
 && \cbThree ${\frak q}=4$ && \cbThree $3 \times 16$ & \cbThree $3 \times 36$ & \cbThree $13[{\rm CX}] + 12[{\rm SWAP}_2] = 49$ & \cbThree {\color{RoyalPurple} $|E| \times 213 - O(N) \times 48$} \\ \cline{2-3} \cline{5-8}
 &&&& $[{\rm SWAP}_d]=3(d-1)$ & \multicolumn{2}{ c|}{$[{\rm CP}]=2(d-1)$} &\\ \cline{2-3} \cline{5-8} 
 & \parbox[t]{2mm}{\multirow{3}{*}{\rotatebox[origin=c]{90}{qudit}}} 
 &  \cbOne ${\frak q}=2$ && \cbOne 9  & \multicolumn{2}{ c|}{\cbOne 6}  & \cbOne {\color{RoyalPurple} $|E| \times {\bf 15} - O(N) \times 9$} \\ \cline{3-3} \cline{5-8} 
 && \cbTwo ${\frak q}=3$ && \cbTwo 21 & \multicolumn{2}{ c|}{\cbTwo 14} & \cbTwo {\color{RoyalPurple} $|E| \times {\bf 35} - O(N) \times 21$} \\ \cline{3-3} \cline{5-8}  
 && \cbThree ${\frak q}=4$ && \cbThree 45 & \multicolumn{2}{ c|}{\cbThree 30} & \cbThree {\color{RoyalPurple} $|E| \times {\bf 75} - O(N) \times 45$} \\ \cline{3-3} \cline{5-8}
 \hline 
 \hline 
 &&& $n^{\rm inter}_{\rm SWAP}$
 & $[\widetilde{\rm SWAP}]$
 & $n^{\rm intra}_{{\rm SWAP}_2} [{\rm SWAP}_2]$
 &
 $\left[ {\rm C}^{{\frak q}-1}(U) \right]$ & ${\cal C}_{\rm tot}$ \\ 
 \hline
 \multirow{7}{*}{2D} & \parbox[t]{2mm}{\multirow{3}{*}{\rotatebox[origin=c]{90}{binary}}} 
 & \crOne ${\frak q}=2$ & \multirow{7}{*}{$O(N)$} & \crOne $3 \times 3\frac{1}{3}^*$ & \crOne $3 \times 4^*$ & \crOne $1 [\rm CX] =1$ & \crOne {\color{RoyalPurple} $|E| \times 17^* + O(N) \times 10^*$} \\ 
 \cline{3-3} \cline{5-8}
 && \crTwo ${\frak q}=3$ && \crTwo $3 \times 8^*$ & \crTwo $3 \times 6^*$ & \crTwo $5[{\rm CX}] = 5$ & \crTwo {\color{RoyalPurple} $|E| \times 29^* + O(N) \times 24^* $} \\ \cline{3-3} \cline{5-8}  
 && \crThree ${\frak q}=4$ && \crThree $3 \times 9\frac{1}{3}^*$ & \crThree $3 \times 12\frac{2}{3}^*$ & \crThree $13[{\rm CX}] + 2[{\rm SWAP}_2] = 19$ & \crThree {\color{RoyalPurple} $|E| \times 65^* + O(N) \times 28^*$} \\ \cline{2-3} \cline{5-8}
 &&&& $[{\rm SWAP}_d]=3(d-1)$ & \multicolumn{2}{ c|}{$[{\rm CP}]=2(d-1)$} &\\ \cline{2-3} \cline{5-8} 
 & \parbox[t]{2mm}{\multirow{3}{*}{\rotatebox[origin=c]{90}{qudit}}} 
 &  \crOne ${\frak q}=2$ && \crOne 9  & \multicolumn{2}{ c|}{\crOne 6}  & \crOne {\color{RoyalPurple} \bf $|E| \times {\bf 6} + O(N) \times 9$} \\ \cline{3-3} \cline{5-8} 
 && \crTwo ${\frak q}=3$ && \crTwo 21 & \multicolumn{2}{ c|}{\crTwo 14} & \crTwo {\color{RoyalPurple} \bf $|E| \times {\bf 14} + O(N) \times 21$} \\ \cline{3-3} \cline{5-8}  
 && \crThree ${\frak q}=4$ && \crThree 45 & \multicolumn{2}{ c|}{\crThree 30} & \crThree {\color{RoyalPurple} \bf $|E| \times {\bf 30} + O(N) \times 45$} \\ \cline{3-3} \cline{5-8} 
 \hline 
\end{tabular}
\end{adjustbox}
\caption{
Gate count as per Eqs. (\ref{eq:Ctot}) for qudits and qubit binary encoding for 1D (blue shaded lines) and 2D (red shaded lines) processor geometries. 
The cost of the $\widetilde{\rm CP}$ gate is evaluated using the circuits from Ref.~\cite{fuchs2021efficient} as well as the standard decomposition of multi-controlled qubit gates shown in the right column of Table~\ref{tab:architectures}.
The cost of the $\widetilde{\rm SWAP}$ gate is evaluated according to the qubit arrangements shown in the left column of Table~\ref{tab:architectures}. In 2D, it is obtained as a weighted average over the neighbours (for ${\frak q}=2,3,4$, each e-dit has four, three and two neighbours to which it is connected by one leg and two, three and four neighbours to which it is connected by three legs).
This might result in a fractional value of $[\widetilde{\rm SWAP}]$ and gate counts stemming from such weighted average are denoted by a star. The total gate count ${\cal C}_{\rm tot}$ of the algorithm is indicated in the last column in purple. For the qubit encoding, ${\cal C}_{\rm tot}$ also includes the contribution $2 {\frak q} [{\rm CX}]$ from (\ref{eq:CP_tilde_cost}) not listed in the Table. For complete graphs considered here, the dominant contribution to ${\cal C}_{\rm tot}$ is coming from the number of edges $|E| = N(N-1)/2 = O(N^2)$. This is highlighted in bold for qudits in the last column. It is apparent from ${\cal C}_{\rm tot}$ that the qudit encoding is superior to the (best-case scenario ${\frak q} = \log_2 d$) binary one in all cases.}
\label{tab:GateCount}
\end{table}

\end{widetext}
%
\subsection{Gate count and comparison to qubits}
\label{sec:PD_Comparison}
%
%
As stated in the introduction, the use of qudits in general offers a resource-efficient alternative to qubit encodings and for certain problems, such as correlation clustering, also a natural physical platform for their implementation. To proceed with the comparison between qudits and qubits, we first specify the qubit encoding. In what follows, we denote by tilde a qudit-like gate acting on the effective qudit encoded in qubits. We also term such an effective qudit an \emph{e-dit} to distinguish it from the physical qudit.

Effective controlled-phase gates $\widetilde{\rm CP}$ equivalent to (\ref{eq:CP_gate}) have been proposed in~\cite{fuchs2021efficient}. 
The Ref.~\cite{fuchs2021efficient} studied the MAX-$k$-cut using QAOA and considered two possible encodings, a one-hot and a binary one. The one-hot encoding seems to produce smaller 2-gate count when ${\frak q} \neq \log_2 d$ and larger when ${\frak q} = \log_2 d$, the precise numbers depending on the graph topology~\cite{fuchs2021efficient}.
Here
\beq
	{\frak q} = \lceil \log_2 d \rceil,
	\label{eq:q}
\eeq
is the number of qubits necessary to contain the qudit Hilbert space as subspace such that the two become equal when ${\frak q} = \log_2 d$. On the other hand, the one-hot encoding is more resource extensive than the binary encoding and we therefore choose the binary one. In the scheme of Ref.~\cite{fuchs2021efficient} ${\frak q} = \log_2 d$ and ${\frak q} \neq \log_2 d$ correspond to two different realizations of $\widetilde{\rm CP}$, which we list in Table~\ref{tab:architectures} for the reader's convenience. 

Next, we have to specify the processor topology. The versatility of neutral atom platforms allows in principle for arbitrary arrangements of the atoms, which can be exploited for efficient encoding of the graph instance at hand. For specificity, we perform the gate count on one of the examples of main interest, the complete graph, for which $|E|=N(N-1)/2$.
Motivated by ongoing experiments \cite{Barredo2016,Endres_2016_Science,Barredo_2018_Nature, scholl2020programmable,semeghini2021probing, Levine_2019_PRL,Madjarov_2020_NatPhys}, we consider a simple 1D chain (with open boundaries) and a 2D regular lattice. In 2D we consider a triangular lattice, which provides the densest packing of atoms.
%
The 2-gate count is given by
\begin{subequations}
	\begin{align}
			{\rm qubits:} \; & {\cal C}_{\rm tot} = n^{\rm inter}_{\widetilde{\rm SWAP}} [\widetilde{\rm SWAP}] + |E| [\widetilde{\rm CP}] \label{eq:Ctot_qubit} \\
			{\rm qudits:} \; & {\cal C}_{\rm tot} = n^{\rm inter}_{{\rm SWAP}_d} [{\rm SWAP}_d] + |E| [{\rm CP}],
			\label{eq:Ctot_qudit}
	\end{align}
	\label{eq:Ctot}
\end{subequations}
where $n^{\rm inter}_{\widetilde{\rm SWAP}}$, $n^{\rm inter}_{{\rm SWAP}_d}$ count the number of SWAPs between the e-dits and qudits respectively such that each vertex has been a neighbour of every other vertex at least once. 
In principle we can perform the gate count for any $d$. We note, that due to the $\widetilde{\rm CP}$ gate structure for ${\frak q} \neq \log_2 d$, the total count is relatively much higher than for ${\frak q} = \log_2 d$ (Table 1 of~\cite{fuchs2021efficient} gives $[\widetilde{\rm CP}]=$\underline{2},70,\underline{6},206,142,78,\underline{14} for $d=$\underline{2},3,\underline{4},5,6,7,\underline{8}, where the values with underline correspond to ${\frak q}=\log_2 d$). By contrast, since the gate structure for qudits is the same for all $d$ and since our main goal here is to compare the gate count for qudits and qubits, we focus only on the ${\frak q} = \log_2 d$ case. The reason for this is that for a given ratio ${\cal C}_{\rm tot}^{\rm qudits}/{\cal C}_{\rm tot}^{\rm qubits}$ of gate counts for qudits and qubits for ${\frak q} = \log_2 d$, this ratio will be only smaller when ${\frak q} \neq \log_2 d$. In this case, the cost of $\widetilde{\rm CP}$ can be further decomposed as
\beq
	[\widetilde{\rm CP}] = n^{\rm intra}_{\rm SWAP_2} [{\rm SWAP}_2] + 2 {\frak q} [{\rm CX}] + [{\rm C}^{{\frak q}-1}(U)],
	\label{eq:CP_tilde_cost}
\eeq
where $n^{\rm intra}_{\rm SWAP_2}$ counts the number of SWAPs within the e-dit and ${\rm C}^{{\frak q}-1}(U)$ is the multicontrolled unitary performed within the target e-dit, cf. Table~\ref{tab:architectures}. The unitary $U$ in ${\rm C}^{{\frak q}-1}(U)$ is a single qubit gate and its particular form is not relevant for the counting (cf. Ref.~\cite{fuchs2021efficient} for details). 

~\\
\emph{1D.} Let us first describe the situation in 1D. Here, $n^{\rm inter}_{\widetilde{\rm SWAP}} = n^{\rm inter}_{{\rm SWAP}_d}$ indicates the number of SWAPs between qudits (or e-dits) such that each qudit (or e-dit) has been a neighbour of all the others. Here we shall use a rather natural choice of SWAP sequences described in~\cite{Kivlichan_2018_PRL,Babbush_2018_PRX}. It consists of a repeated application of two layers of SWAPs, one performing SWAP operations on qudit (e-dit) pairs $1-2, 3-4,\ldots$ followed by the other on pairs $2-3,4-5,\ldots$. This yields $n^{\rm inter}_{\widetilde{\rm SWAP}} = n^{\rm inter}_{{\rm SWAP}_d} = (N-1)(N-2)/2$, see also Appendix~\ref{app:HW}. Furthermore we also have $|E|=N(N-1)/2$ such that Eq.~(\ref{eq:Ctot_qudit}) together with Eq.~(\ref{eq:cost_CP_Kj}) and Eq.~(\ref{eq:cost_SWAPd}) yield
\beq
	{\cal C}_{\rm tot} = \frac{(5N-6)(N-1)}{2}(d-1) = 
	|E| \times 5(d-1) - O(N).
\eeq

To analyse the qubit case, let us start with the analysis of the ${\rm C}^{\frak q-1}(U)$ gate. Ref.~\cite{Liu_2007} describes a systematic construction of ${\rm C}^{\frak q}(U)$ for arbitrary $\frak q$ proposing two schemes for such a construction, one scaling exponentially and the other one polynomially with $\frak q$. While the polynomial scaling is clearly favourable for large $\frak q$, the actual gate count favours the exponential scheme for $\frak q \leq 4$ considered here. We thus consider the decomposition for ${\rm C}^2(U)$ and ${\rm C}^3(U)$ shown in Table~\ref{tab:architectures}, proposed in the early works on quantum information~\cite{Barenco_1995_PRA, Sleator_1995_PRL}. The considered ordering within the e-dit shown in Table~\ref{tab:architectures} leads to the following counts as function of ${\frak q}$
\begin{subequations}
	\begin{align}
		[\widetilde{\rm SWAP}] &={\frak q}^2 [{\rm SWAP_{2}}] \\
		n^{\rm intra}_{\rm SWAP_2} &=3 {\frak q}({\frak q} - 1).
	\end{align}
\end{subequations}
The total gate count for qudits and e-dits and their break down as per (\ref{eq:Ctot}) is summarized in Table~\ref{tab:GateCount}.

~\\
\emph{2D.} For qudits, the cost (\ref{eq:Ctot_qudit}) carries over to two dimensions. For qubits, the situation is more involved and we shall analyse only the specific cases ${\frak q}=2,3,4$. In Table~\ref{tab:architectures} we show a possible arrangement of the e-dits including the qubit labelling within the e-dit. We note that the effective lattice geometry composed of the e-dits retains the topology of the triangular lattice in that each e-dit has six nearest neighbours and consequently $n^{\rm inter}_{{\rm SWAP}_2} = n^{\rm inter}_{{\rm SWAP}_d}$ as in the 1D case. However the difference with qudits is that the e-dit lattice is ``anisotropic'', namely for ${\frak q}=2,3,4$, each e-dit has four, three and two neighbours to which it is connected by one leg and two, three and four neighbours to which it is connected by three legs respectively.

We also note that the proposed tilings implementing the e-dits are not necessarily unique. 

In order to determine $n^{\rm inter}_{{\rm SWAP}}$, one can apply a generalization of the alternating SWAP sequence from the 1D case, cf.~\cite{Babbush_2018_PRX}, which yields a scaling $O(N)$ for the number of SWAPs between the qudits (or e-dits). Since $|E|=N(N-1)/2 = O(N^2)$, this gives a subleading contribution to the gate count and we do not elaborate on the precise sequence further.

We are thus left with evaluating the cost of $[\widetilde{\rm SWAP}]$ for the e-dit SWAP and $[\widetilde{\rm CP}]$ which takes into account the respective e-dit and processor geometries. The summary of the costs for 2D is given in Table~\ref{tab:GateCount}.

As a result, we find the expected outcome, namely that in all considered cases the gate count is lower for qudit encoding than for qubit binary encoding, even for the best-case scenario ${\frak q} = \log_2 d$ of the latter.

Here we have evaluated the gate count considering the realization of the SWAP gates via the Rydberg interactions. A remark is that the gate count of a sequence of consecutive SWAPs can be further reduced by considering so-called bridge gates, leading however only to a modest improvement by a factor $\approx 1.5$~\cite{Itoko_2020}. In the context of neutral atoms in optical tweezers, it would be interesting to exploit a strength of these platforms and perform the SWAP by physically exchanging the atoms, which for distances $\sim 5\, \mu{\rm m}$ can be done on the timescale of $\sim 50\,\mu{\rm s}$~\cite{Schymik_2020_PRA}.



\section{\label{sec:Errors} Errors}

In this section we consider an error model used in Ref.~\cite{Auger_2017_PRA} in the theoretical analysis of a Rydberg quantum computer and we discuss the implications of the errors for the algorithm performance. 
Importantly, the error model of Ref.~\cite{Auger_2017_PRA} can be cast in the \emph{unitary} evolution framework used in our work. We comment on the actual experimental errors and how they relate to the considered error model in Appendix~\ref{app:Errors}.

~\\
\noindent
\emph{Unitary error model.}
Let us consider a set of $d^2$ single qudit unitaries ${\mathbb U} \equiv \{ (\Sigma^X)^r (\Sigma^Z)^s\}$, where $r,s = 0,\ldots,d-1$. Here $\Sigma^X = \Sigma^x + (\Sigma^x)^\dag$, $\Sigma^x$ is given by Eq.~(\ref{eq:Sigma_x}) and
\beq
\Sigma^Z = 
\begin{bmatrix}
1 & 0 & \ldots & 0 \\
0 & \lambda & & \vdots \\
\vdots & & \ddots & 0 \\
0 & \ldots & 0 & \lambda^{d-1}
\end{bmatrix}
\eeq
with $\lambda = {\rm exp}(i 2\pi/d)$ is the generalized Pauli Z.
For a qubit, $d=2$, this reduces to ${\mathbb U} = \{{\mathds 1}, \Sigma^X, \Sigma^Z, \Sigma^X \Sigma^Z \}$. Motivated by the experimental considerations, namely the fact that the errors are dominated by the 2-gate ones, cf.\  Appendix~\ref{app:Errors}, the model consists of applying an identity with probability $1-p_2$ or a unitary $U \in {\mathbb U}^{\otimes 2} \setminus \{ \mathds{1} \}$, 
with probability $p_2/(|{\mathbb U}^{\otimes 2} \setminus \{ \mathds 1 \}|) = p_2/(|{\mathbb U}^{\otimes 2}|-1)$ on each pair of qudits after each cost unitary $U_C = {\rm e}^{-i H_C}$, cf.\ Eq.\ (\ref{eq:HC}). Put formally, for $\rho =  U_C \rho' U_C^\dag$
\beqa
    \rho &\rightarrow& \rho \;\; {\rm with \; prob.}\;1-p_2 \nonumber \\
    \rho &\rightarrow& U \rho U^\dag, \; U \in {\mathbb U}^{\otimes 2} \setminus \{ \mathds{1} \}, \;\; {\rm with \; prob.}\;\frac{p_2}{|{\mathbb U}^{\otimes 2} \setminus \{ \mathds 1 \} |} \nonumber \\
    \label{eq:err_model}
\eeqa
We consider the same data sets for complete graphs as have been used in Section \ref{sec:P} for $N\in\{3,4,5\}$, and use the \emph{optimized} values for $\gamma,\beta$ and $d$ obtained in the noise-free setting. This way makes it possible to discard the classical optimization loop, saving computation time, and allows us to focus on the performance degradation as a result of the randomization of the state vector. It has been argued in Ref.~\cite{Xue2019} that noise generated by dephasing, bit flip, and depolarizing channels tends to flatten (on average) the parameter space energy landscape without changing its structure. Since the error model (\ref{eq:err_model}) is a qudit generalization of these types of channels, we expect the $\gamma,\beta$ obtained in the noise-free setting to be optimal also in the noisy setting.

The results are shown in Figure~\ref{fig:error_results}. For $p_2$ small ($ \lesssim 10^{-3}$), we see that the performance is hardly affected by the noise. Once $p_2$ increases, we enter the regime where performance quickly degrades until we reach the performance of the completely randomized state. Whilst the  performance at $p_2=1$ is considerably smaller than in the noise-free setting, the approximation ratio achieved on average is still relatively large. This is due to the fact that $d$ is already pre-determined: instances with $d=1$ are not affected by the noise and still maintain a large approximation ratio. We define the threshold noise~$p_{2,\text{Th}}$ (threshold amount of 2-gates $g_\text{Th}$) as the noise level (amount of 2-gates) for which the QAOA has lost half of its performance as compared to random guessing, on average for all instances. 

By determining the values of $p_{2,\text{Th}}$ from the data in Figures~\ref{fig:error_results}a-c, and using that the amount of two-qubit gates is simply $pN(N-1)/2$ (with $p$ the QAOA depth), we plot $g_\text{Th}$ as a function of $p_2$ in Figure~\ref{fig:error_results}d. We observe that our data are compatible with a linear dependence for $g_\text{Th}$ of the form
\beq
    g_{\text{Th}}=\frac{\kappa}{p_2}, \;\; \kappa=0.84.
    \label{eq:Error_scaling}
\eeq
This naive model for the noise shows scaling similar to that of Ref.~\cite{Franca2020} -- here the authors showed that a bound on the circuit depth scales inversely proportional to the quantum gate error -- but now in the amount of operations instead of circuit depth.

It is interesting to interpret the results of Fig.~\ref{fig:error_results}d in terms of the achievable system sizes and required hardware gate operations. Denoting the error probability of an elementary CX-like hardware gate as $p_{\rm CX}$ we get for the total success probability (no error) after application of $|E|=g$ gates on the $|E|$ edges
\beqa
    p_{\rm success} &=& (1-p_2)^{|E|} \nonumber \\
    &=& (1-[{\rm CP}] p_{\rm CX})^{|E|} (1-[{\rm SWAP}_d] p_{\rm CX})^{O(N)} \nonumber \\
    &\approx & (1-[{\rm CP}] p_{\rm CX})^{|E|},
    \nonumber
\eeqa
where in the last relation we used that $p_2, p_{\rm CX} \ll 1$, $|E| \gg N$ while $[{\rm CP}] \approx [{\rm SWAP}_d]$. This allows us to identify $p_2 \approx [{\rm CP}] p_{\rm CX}$. Using the result of Table~\ref{tab:GateCount}, considering $d=8$ for specificity, and using $|E|=N(N-1)/2$, we find that for $p_2 = (10^{-2},10^{-3},10^{-4})$, $p_{\rm CX} \approx (10^{-3},10^{-4},10^{-5})$ resulting in $N \approx 13, 41, 130$
\footnote{A technical remark is that strictly speaking the formula (\ref{eq:Error_scaling}) was verified when the number of clusters $d$ saturates the number of vertices $N$. Keeping this in mind, the discussed achievable system sizes should thus be considered merely as an estimate.
}.

We thus see that the inclusion of the errors strongly limits the scaling of the algorithm to large problem instances. We note that similar limitations of the QAOA due to experimental errors have been discussed recently in~\cite{harrigan2021quantum} and~\cite{Franca2020}. 
In particular, it has been argued in Refs.~\cite{Franca2020,harrigan2021quantum} that hardware-native graph instances, in case of Ref.~\cite{harrigan2021quantum} a simple square lattice with nearest-neighbour interactions, are less prone to errors and potentially allow to scale up to problem sizes large enough to achieve quantum advantage.
This is due to significant simplification of the quantum circuit, which avoids a number of extra compilation steps (such as SWAP gates).
The major downside is however that a quantum computer hardware is, typically, not application or instance specific.

In this respect, the neutral atom based platforms seem to be particularly interesting as they allow for implementing graph topologies beyond simple planar ones by exploiting the atom arrangements in three dimensions and the long-range interactions. Nevertheless, as is apparent from the results in this section, the robustness to the errors is directly related to the graph topology. Specifically, for the complete graph considered here, the errors strongly limit the scaling even if it is native to the hardware, see also the discussion in Sec.~\ref{sec:Conclusions}.

\begin{figure}
\centering
    \includegraphics[width=\linewidth]{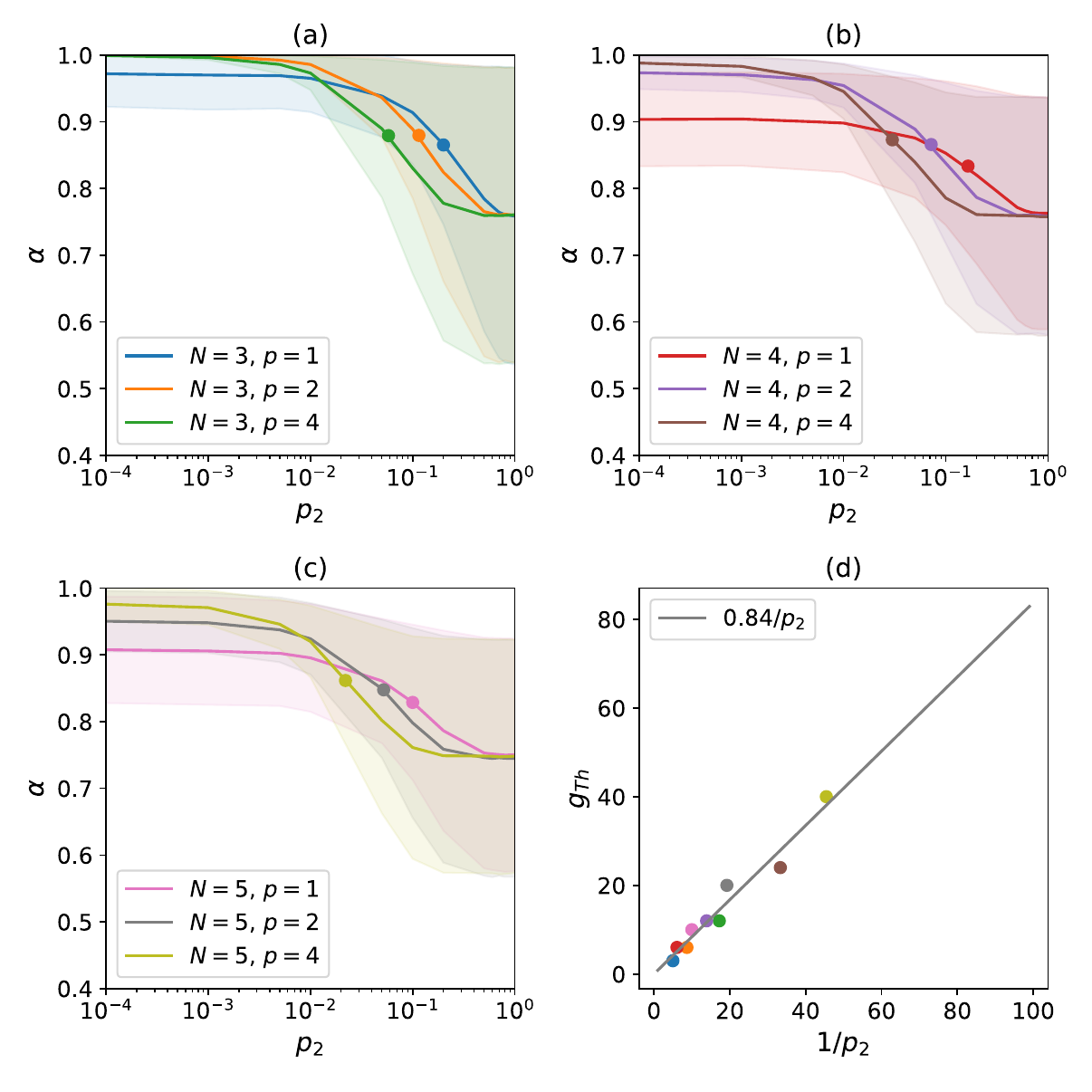}
    \caption{{\bf (a)-(c)} Numerical results for the QAOA with a error channel given by~\eqref{eq:err_model} for $N=3,4,5$, respectively. The plots are the average approximation ratios $\alpha$ over all instances in the complete graph data sets and the shaded areas correspond to one standard deviation. The filled circles data points correspond to the threshold noise for which the QAOA has lost half of its performance over random guessing. {\bf (d)} Linear fit for $g_{\rm Th} \propto 1/p_2$, individual colours correspond to the respective data points where $p_2=p_{2,\text{Th}}$ in (a), (b) and (c).}
    \label{fig:error_results}
\end{figure}



\section{\label{sec:Conclusions} Conclusions and outlook}

In this work we have addressed the problem of solving the correlation clustering using QAOA and a qudit quantum computer. We have specifically considered a neutral atom based architecture, which has the potential to offer up to $\sim$1000 qudits in the near future~\cite{scholl2020programmable,semeghini2021probing,Morgado2020}. Here the gates are realized through the interaction of atoms in a highly excited electronic state, a so-called Rydberg state. Considering specifically the element $\;{}^{87}{\rm Sr}$ we have identified a suitable level structure for the qudit, which in turn allowed us to design the gates for implementing the QAOA on the quantum processor. It is worth emphasizing that while we have considered correlation clustering, the discussed qudit gates can be used for the closely related MAX-$k$-CUT and MAX $k$-VERTEX COVER problems, cf. \cite{fuchs2021efficient} and~\cite{cook2020quantum,bartschi2020grover,wang2020preparation} respectively.

We assessed the algorithm's performance by numerical simulations including various optimization strategies, namely restarts, optimizing the initial points and looping over the cluster number. Focusing specifically on complete and Erd\H{o}s-R\'enyi graphs of up to 7 vertices and clusters we found that in all studied cases the QAOA with depth $p\geq 2$ provides approximation ratios above the Swamy bound 0.7666 \cite{Swamy2004}, corresponding to the best known classical strategy (based on SDP) with performance guarantees for MAXAGREE. Modifying and adopting recently developed classical simulation methods \cite{2021npjQI...7..101M} might allow us to study larger graph instances, which we leave to future work.

While this result is encouraging, the inclusion of errors suggests that it is challenging for the QAOA to outperform classical algorithms, cf. Sec.~\ref{sec:Int}, at least for complete graphs on near-term noisy quantum devices. This is in agreement with related results reported recently in Ref.~\cite{Franca2020} and Ref.~\cite{harrigan2021quantum}, where various graph instances were considered to compare the performance of the QAOA using a superconducting quantum chip, including the effect of the errors. 


In this respect, the neutral atom based platforms seem to be particularly interesting and our work indicate possible directions in the design of experiments for benchmarking the QAOA on correlation clustering instances.
Not only do neutral atoms allow to assemble arbitrary structures in both 2D and 3D~\cite{Barredo2016,Endres_2016_Science,Barredo_2018_Nature,Schymik_2020_PRA} (a PTAS does exist for planar graphs~\cite{Klein2015}, but one could still investigate the rate of convergence), but they also allow for native long-range interactions, known to yield an advantage over quantum processors featuring only nearest-neighbour ones~\cite{linke2017experimental}. Such long-range interactions in turn allow for implementation of \emph{non-planar} graphs~\cite{harrigan2021quantum} as native geometry using even a planar arrangement of atoms and at the same time avoid the need for additional SWAP gates.
One should however keep in mind that for complete graphs, the errors strongly limit the scaling even when this graph topology is native to the hardware, cf. Sec.~\ref{sec:Errors}. This suggests, together with the results of Refs.~\cite{Franca2020,harrigan2021quantum}, that the robustness of (low depth) QAOA increases with the decrease of the graph degree. Additional improvement in the performance might be also achieved when using multiqudit gates, which can further reduce the circuit depth~\cite{Kiktenko_2020_PRA}.
In this context, it would be also interesting to consider the (native) realizations of unit disk graph instances~\cite{2018arXiv180810816P,2018arXiv180904954P}. Interestingly, and to the best of our knowledge, the NP-hardness of the correlation clustering problem on unit disk graphs remains an open question~\cite{Daz2007MAXCUTAM,10.1007/978-3-319-44914-2_3}.

It would be highly interesting to address systematically the above listed scenarios, which we leave for future work. It should be also emphasized that while achieving a practical quantum advantage using the QAOA remains a challenge, constructing a qudit quantum processor is a task worth pursuing - it constitutes an exquisite tool for applications beyond the QAOA, allowing for instance for the realization of a plethora of condensed matter models such as the $d-$state Potts and other SU($d$) spin systems~\cite{Wu_1982_RMP, nataf2014exact, nataf2016exact, kim2017linear, nataf2018density,dufour2015variational,romen2020structure,chen2015quantum,d2015renyi,song2013mott,Corboz_2012_PRB,Corboz_2013_PRB,Nataf_2016_PRB,Corboz_2011_PRL,Bauer_2012_PRB,weichselbaum2018unified}.

    


\section*{Acknowledgements}
We thank Alexander Brinkman for his supervision of the master's project~\cite{Weggemans2020} which led to this publication and Koen Groenland for useful discussions. This work was carried out on the Dutch national e-infrastructure with the support of SURF Cooperative and supported by the Dutch Ministry of Economic Affairs and Climate Policy (EZK), as part of the Quantum Delta NL programme
and by the Netherlands Organization for Scientific Research (NWO) under the Gravitation grant No. 024.003.037 and the Quantum Software Consortium.
AR acknowledges the support of the German Federal Ministry of Education and Research in the funding program “quantum technologies – from basic research to market” (contract number 13N15585).

~\\
\noindent
{\bf Author contributions.}
JW and FSp have performed the analysis and numerical simulations of the algorithm. AR has supported the analysis of the numerical simulations. RB encouraged JW to investigate performance guarantees and supervised the findings of this work.  JM and KS have designed the gates and with AU, RS and FSc have proposed and analysed the hardware implementation. 
All authors have contributed to the writing of the manuscript.



\bibliographystyle{plainnat}
\bibliography{main.bib}


\appendix

\onecolumngrid

\section{\label{app:optstudy}Results optimizer study}
We compare the performance of five different optimizers with off-the-shelf hyper-parameter settings using the vanilla QAOA: ImFil, SnobFit and BOBYQA taken from the scikit-quant package~\cite{Lavrijsen2020} and an implementation of COBYLA, used both as single optimizer and in conjunction with basin-hopping (BH), both from the SciPy optimization package~\cite{Scipy}.  We consider a single instance out of our correlation clustering data sets: the complete graph with $N=4$ and all edge weights `$-$' such that the optimal solution corresponds to all nodes being put in different clusters. For each $p \in \{1,\dots,5\}$, we generate 25 random initial points in the respective parameter space. The maximum number of iterations for each optimizer is set to 500. Since basin-hopping has two different budgets, one for the basin-hopping steps and one for its local optimizer, we set the maximum number of iterations by using the number of evaluations the local optimizer (COBYLA) used in its individual run: the number of basin-hopping steps is the rounded ratio of the budget over this number. The results for state vector sampling ($1000$ samples) are given in Figure~\ref{fig:optimizer_study}. 
\begin{figure}[H]
    \centering
    \includegraphics[width=\linewidth]{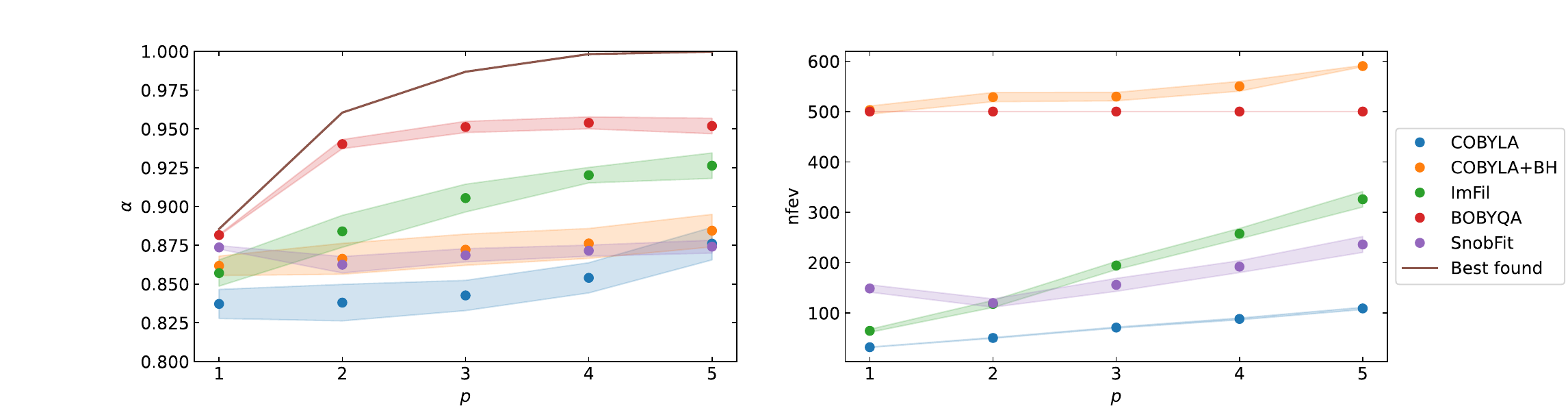}
    \caption{State vector sampling with 1000 samples. Left: found approximation ratios $\alpha$  for 25 random initial points using different optimizers. Right: total number of function evaluations `nfev' as a function of $p$. The shaded area indicates the error in the mean, where the discrete points have been connected in order to improve the readability of the figure. 'Best found' indicates the best approximation ratio that was observed over all instances and optimizers.}
    \label{fig:optimizer_study}
\end{figure}
We find that BOBYQA outperforms all other optimizers in terms of the achieved approximation ratios, but does always use up the maximum available amount of iterations. Unfortunately, the scikit-quant optimization package does not allow us to change the tolerance, which would allow for fairer comparison. As far as the ratio of performance to number of function evaluations is concerned, ImFil performs well. Even though adopting BOBYQA already potentially results in a large performance increase compared to using COBYLA, which was used to obtain our initial results, we still observe that a large performance difference exists between the best value found and the average performance over different optimizer runs. In addition, note that the standard deviation is relatively large as we plot the error in the mean in Fig.~\ref{fig:optimizer_study}. This means that there is still a lot to be gained in the classical optimization step, the most natural being the use of good initial points, followed by hyper-parameter optimization. Good initial points will have a larger effect on the local optimization methods (in particular COBYLA and BOBYQA) compared to the global optimizers (e.g.\ ImFil). The concentration of initial points at $p=1$ is studied in the next appendix. 

\section{\label{app:ipstudy}Initial points study for $p=1$}
 In this part of the appendix we give numerical evidence that supports the conclusions of the work by Brandão et al.~\cite{Brandao2018}: initial points for different instances, belonging to similar problem classes, are concentrated. To investigate this effect at $p=1$, we run the following: we start with initial points we obtained from the all-negative-weights graphs, for some fixed $N$ and $d$, and then solve for all 50 instances in our correlation clustering data set, using COBYLA as the optimizer. The resulting optimal points are plotted in Fig.~\ref{fig:points}:
\begin{figure}
    \centering
    \includegraphics[width = 0.9\linewidth]{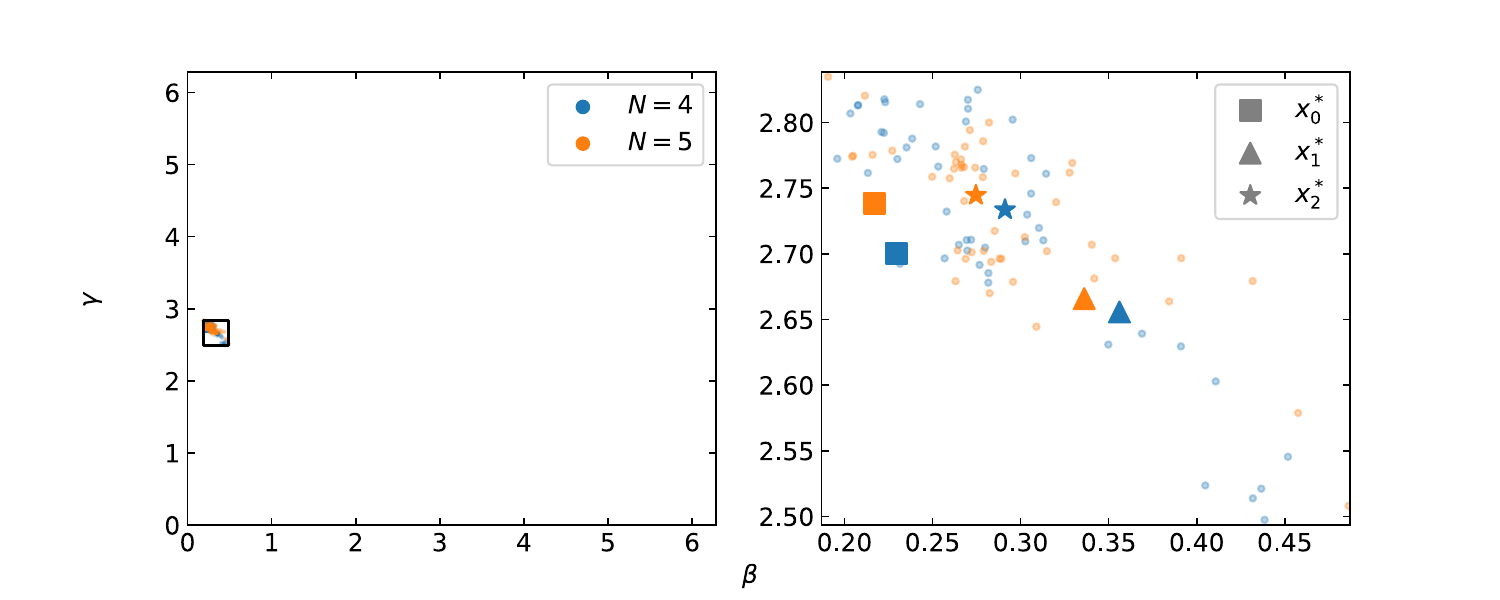}
    \caption{Locations of points obtained through the use of COBYLA starting from initial point $x_0^*$, which corresponds to an optimized point for solving the all-negative-weights graph, for $N=4$ and $N=5$. The number of levels is set as $d=N$. Left:  overview over the entire possible parameter space. Right: close-up to the smallest possible square area that contains all optimal points. The point $x_1^*$ has the smallest maximum Euclidean distance to other points and  $x_2^*$ the smallest average Euclidean distance to all other points.}
    \label{fig:points}
\end{figure}
Note how all points are in the neighbourhood of our initial points. In fact, the smallest rectangular area containing all points (indicated by the black rectangle) takes about $0.2\%$ of the entire possible parameter space. The plot on the right zooms in on this rectangle, showing how the optimal points are located relative to each other. In this plot $x^*_0$ is the initial point, which is for both $N=4$ and $N=5$ positioned at the boundary of the collection of points. This makes sense due to the structure of the graph it belongs to---the all-negative-weights graph is itself an extreme case of the correlation clustering problem (requiring all clusters to be used). This also indicates that other graphs might be more suitable to generate initial points. 

\begin{figure}
    \centering
    \includegraphics[width = 0.9\linewidth]{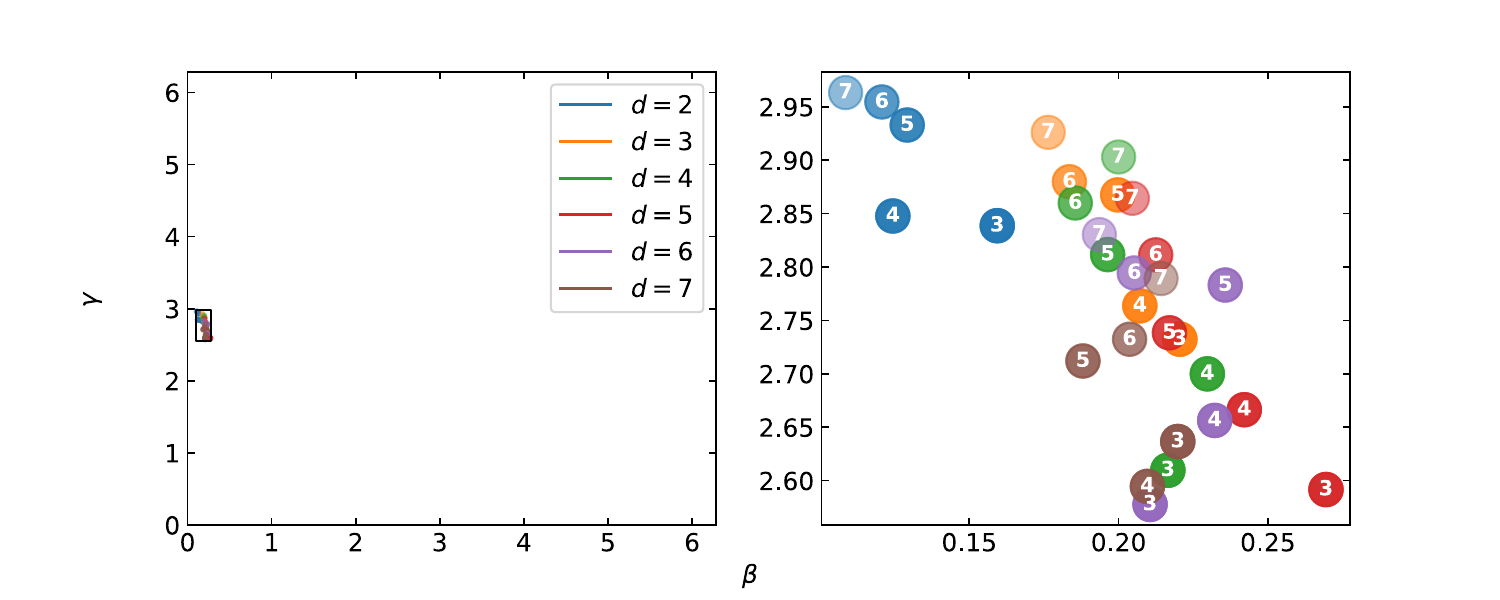}
    \caption{Left: locations of optimal points obtained over the entire possible parameter space. Right: close-up of the smallest rectangle that contains all optimal points. The number inside the point indicates the number of nodes.}
    \label{fig:points2}
\end{figure}
For the same problem instance, Figure~\ref{fig:points2} shows the location of the points obtained for different $N$ and $d$ in parameter space ($p=1$). We observe that $d=2$ is somewhat of an outsider, but again all points fall in a rectangular area encompassing about $0.2\%$ of the entire possible parameter space.

\section{\label{app:app_ratio_bound}Approximation ratio bound on 3-regular graphs}
A recent work by Wurtz and Love~\cite{Wurtz2020} shows a derivation of lower bounds for QAOA depths $p=1$ and $p=2$ (and a conjectured result at $p=3$) on MAXCUT. In this appendix, we apply their techniques to the correlation clustering problem on 3-regular graphs: graphs for which every node has a fixed degree of 3. The goal is to find a lower bound for the approximation ratio $\alpha$ at $p=1$ on 3-regular graphs, defined as
\begin{equation}
    \alpha = \max_{d \in \{1,\dots,d_\text{max}\}} \frac{F_d(\gamma,\beta)}{C_\text{max}},
\end{equation}
where $F_d(\gamma,\beta)$ is the expectation value of the state produced by the QAOA algorithm using $d$ levels and $C_\text{max}$ denotes the optimal objective function value for a single correlation clustering instance. Since only graphs with degree 3 are considered, we need at most 4 clusters. Therefore, we only consider $d \in \{1,2,3,4\}$. Based on local optimization, we will conjecture that for 3-regular graphs $G=(V,E)$ initial points $\gamma^*_d,\beta^*_d$ exist such that QAOA that loops over the clusters has an approximation ratio larger then $0.6699$, even without the classical optimization loop\footnote{Our derivation does not include the possibility that $G$ is the complete $N=4$ graph, which is not neccesary as it cannot be extended.}. By relaxing the problem into a linear program (LP) we can prove that this bound is at least $0.6367$.

\subsection{Problem setup and a lower bound for the energy}
Consider an arbitrary 3-regular graph $G$ of $N(G)$ nodes that is not the complete graph with $N=4$. We identify for each edge $\langle i,j \rangle$ the sub-graph $G^p_{<i,j>}$ induced by all neighbouring  edges at most $p$ steps away from $\langle i,j \rangle$. At $p=1$ there are only three possible kinds of sub-graph structures as indicated in Fig.~\ref{fig:subgraphs_p1}.
\begin{figure}[H]
    \centering
    \includegraphics[width=0.7\linewidth]{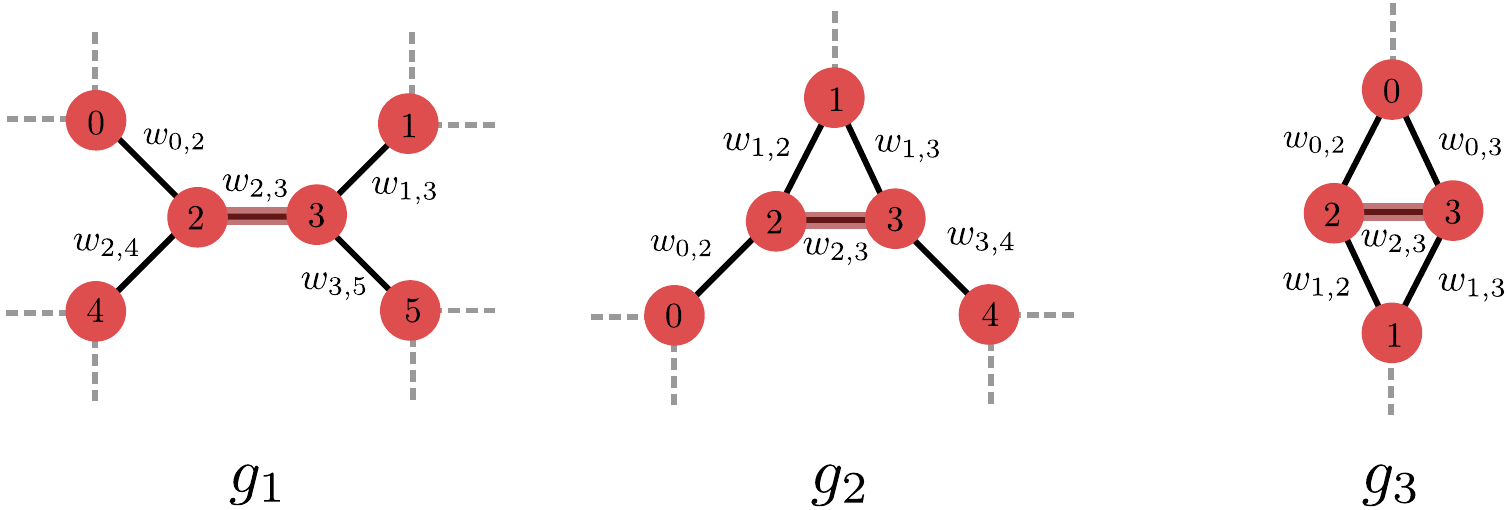}
    \caption{The 3 types of sub-graphs for $p=1$. The sub-graphs describe the environment of the highlighted edge, note how only neighbouring edges are included in the sub-graph. The dotted edges indicate edges outside the sub-graph.}
    \label{fig:subgraphs_p1}
\end{figure}
Since all sub-graph types have 5 edges, we have a total of $3\cdot 2^5=96$ possible sub-graphs when we only consider weights $w_{u,v} \in \{-1,+1\}$. However, by symmetry arguments we can reduce the total number of weighted sub-graphs we have to consider: if reordering the edge labels results in the same graph under any rotation, it will look the same to the QAOA. We define three sets of sub-graphs $g_i, i \in \{ 1,2,3 \}$, representing all 3-regular sub-graph structures with $6$, $5$ and $4$ nodes respectively (see Figure~\ref{fig:subgraphs_p1}), such that for every sub-graph $\lambda \in g_i$ there exists no other graph $\lambda' \in g_i$ that is equivalent in the QAOA setting. Our total set of possible sub-graphs is then $S = g_1 \cup g_2 \cup g_3 $. We now decompose our graph $G$ into sub-graph environments $\lambda \in S$ for which the multiplicity of each $\lambda$ is $N_\lambda(G)$. Since we have such a sub-graph environment for every edge, we must have that the sum over all $N_\lambda (G)$ is equal to the total number of edges $3N(G)/2$. For a single edge $\langle u,v \rangle$ with sub-graph environment $\lambda$, denoted as $\langle u,v \rangle\rightarrow \lambda$, the contribution to the expectation value is given by
\begin{equation}
    \begin{aligned}
        f_{d,\langle u,v \rangle\rightarrow \lambda} (\beta_{d},\gamma_{d}) = \bra{\beta_{d},\gamma_{d}} w_{\langle u,v \rangle} V_d \ket{\beta_{d},\gamma_{d}}.
        \label{eq:f_subgraph}
    \end{aligned}
\end{equation}
Note that Eq.~\eqref{eq:f_subgraph}  only contains terms that operate within $\lambda$. The expectation value of the algorithm with a maximum of $k$ clusters at $p=1$  for some $\beta_{d},\gamma_{d}$ is given by
\begin{equation}
    F_{d}(\beta_{d},\gamma_{d}) = \sum_{\lambda} N_\lambda (G) f_{d,\langle u,v \rangle\rightarrow \lambda} (\beta_{d},\gamma_{d}).
\end{equation}
We must also have that this is always smaller than or equal to the expectation value using the best values of $\beta_{d},\gamma_{d}$
\begin{equation}
    F_{d,\text{max}} = \max_{\beta_{d},\gamma_{d}} F_{d}(\beta_{d},\gamma_{d}) \geq \sum_{\lambda} N_\lambda (G) f_{d,\langle u,v \rangle\rightarrow \lambda} (\beta_{d},\gamma_{d}),
\end{equation}
providing us with a lower bound on the expectation value of the algorithm $F_{d,\text{max}}$ at a given $d$.
\subsection{Upper bound on the number of agreements}
We are now faced with the task of finding an upper bound on $C_\text{max}$. A naive bound would be the total number of edges, but we can do better by considering the same argument Wurtz and Love used to determine an upper bound for MAXCUT~\cite{Wurtz2020}. Consider the graph $\mathcal{G}$, which is a collection of $N_\lambda(G)$ disconnected sub-graphs $G^{p}_{\langle i,j\rangle}$ for each edge in $G$. Since the largest sub-graph has only $6^4$ feasible solutions, a brute-force method can be used to find the ratio between the optimal objective function value and the number of edges for every sub-graph $\lambda$, which we will call $c_\lambda$. Since all sub-graphs are isolated (i.e. not connected to each other), the global fraction of agreements to edges is equal to the average fraction over each sub-graph. However, we have several edges belonging to different disjoint sub-graphs in $\mathcal{G}$ that are actually the same edge in $G$. In this case, we can have that for both sub-graphs a clustering exists for which the edge contributes to the objective value, but that the required clustering is different in both sub-graphs. As a result, we have that the objective function value of $G$ is bounded from above by
\begin{equation}
    C_\text{max} \leq \sum_{\lambda} N_\lambda (G) c_\lambda.
    \label{eq:upper_bound1}
\end{equation}

\subsection{Constructing the hardest graph}
Since we have a lower bound for the algorithm's expectation value and an upper bound for the optimal objective function value, the approximation ratio $\alpha$ is bounded from below by
\begin{equation}
    \alpha(G) \geq \max_{d{}}  \frac{\sum_{\lambda} N_\lambda (G) f_{d,\lambda}(\gamma_{d},\beta_{d})}{\sum_{\lambda} N_\lambda (G) c_\lambda}.
    \label{eq:r}
\end{equation}
The worst case approximation ratio is then given by the hardest graph $G = G^{*}$, which corresponds to some combinations of $N_\lambda (G^{*}) \in \mathbb{Z}_0$ for all sub-graphs $\lambda$. However, not all combinations of $N_\lambda (G)$ correspond to a valid graph, as was already noted by Farhi et al.~\cite{Farhi2014}. For now, we will only consider the structure of the graph and not take the feasibility of certain weight combinations into account. First, we note that for every edge in sub-graph $g_3$  there must be at least 4 other edges that have the environment corresponding to sub-graph $g_2$. Also, we have that the `triangle' of $g_2$ and `crossed square' of $g_3$ cannot share the same vertex, which means that the number of triangular edges and crossed square edges must be smaller than the number of nodes $N(G)$. Our final constraints are that all $N_\lambda (G)$ are non-negative integers and must sum up to $3N(G)/2$.  We define $n_\lambda (G) \equiv N_\lambda(G)/N(G)$ such that we can relax the integer constraint in the limit of large $N(G)$, and obtain the following \textit{Minimax Linear Fractional Program}:
\begin{equation}
\begin{aligned}
    &\min_{n_\lambda(G): \lambda \in S} &&\max_{d{}}  \frac{\sum_{\lambda} n_\lambda (G) f_{d,\lambda}(\gamma_{d},\beta_{d})}{\sum_{\lambda} n_\lambda (G) c_\lambda}\\
     &\text{s.t.} \quad     && \sum_{n_\lambda (G):\lambda \in g_2} n_\lambda (G) - 4 \sum_{n_\lambda (G):\lambda \in g_3} n_\lambda (G) \geq 0\\
   &  &&- \sum_{n_\lambda (G): \lambda \in g_2} n_\lambda (G)\geq -1 \\
    & &&\sum_{n_\lambda (G): \lambda \in S} n_\lambda (G) = \frac{3}{2} \\
     & &&n_\lambda (G) \geq 0 \text{ for all } \lambda \in S
\end{aligned}
\label{eq:minimax}
\end{equation}
Equation~\eqref{eq:minimax} is in the form of a generalized fractional program, which is not reducible to a linear program (LP) which can be solved efficiently. However, there exist linear relaxation bounding techniques that do allow for global optimization up to some error $\epsilon$. We have performed initial experiments with one of those techniques~\cite{Jiao2014}, but were not able to achieve desirable results so far due to the difficulties in approximating our objective function. However, we can construct two related LP formulations that upper and lower bound the value of $\alpha$. Let us define the feasible region $\mathcal{C}$ such that $x \in \mathcal{C}$ when it satisfies the constraints of~\eqref{eq:minimax}.
\begin{enumerate}
    \item Take the number of edges as the upper bound instead of the fractional objectives~\eqref{eq:upper_bound1}, i.e.\ we let $c_\lambda \rightarrow 1$ for all $\lambda.$  Under this relaxation we can write~\eqref{eq:minimax} as the following LP:
\begin{equation}
\begin{aligned}
    &\min && \alpha \\
     &\text{s.t.} \quad &&\sum_{\lambda} n_\lambda (G) f_{d,\lambda}(\gamma_{d},\beta_{d}) \leq \frac{3}{2}\alpha \quad \text{for all } d\\
     & &&n_\lambda (G) \in \mathcal{C}, \alpha \in \mathbb{R} 
\end{aligned}
\label{eq:minimaxLP1}
\end{equation}
The solution of~\eqref{eq:minimaxLP1} is a lower bound to the actual bound, as sub-graphs that originally had $c_\lambda=0.8$ now contribute too much to the upper bound of the optimal objective function value~\eqref{eq:upper_bound1}.
    \item Similarly, we can also assume that only sub-graphs for which a perfect clustering exists contribute to the construction of the most difficult graph. Define the set $S'=\{\lambda|\lambda \in S, c_\lambda=1\}$ to be such a set.  Since $S' \subset S$, the resulting solution provides an upper bound to the best value of $\alpha$ that we can find for $\lambda \in S$.
    \begin{equation}
\begin{aligned}
    &\min && \alpha \\
     &\text{s.t.} \quad &&\sum_{\lambda} n_\lambda (G) f_{d,\lambda}(\gamma_{d},\beta_{d}) \leq \frac{3}{2}\alpha  \quad \text{for all } d\\
     & &&n_\lambda (G) \in \mathcal{C}, \alpha \in \mathbb{R},\lambda \in S' 
\end{aligned}
\label{eq:minimaxLP2}
\end{equation}

\end{enumerate}
All LPs will be solved by CVXOPT contained in the package \textit{lpsolvers}~\cite{lpsolvers2007} for Python.

\subsection{Iterative procedure for determining the bound}
The minimax optimization problem gives us a method to determine the hardest instance $G^{*}$, given that we fix the parameters $\gamma=(\gamma_1,\gamma_2,\gamma_3,\gamma_4),\beta = (\beta_1,\beta_2,\beta_3,\beta_4)$. But how do we choose the values of these parameters? As one would normally do with QAOA, a classical optimization loop can be adopted. First, we choose some initial values for $\gamma,\beta$ and calculate $f_{d,\lambda}(\gamma_d,\beta_d)$ for all sub-graph environments $\lambda$ and all $d\in\{1,2,3,4\}$. Next, we construct the hardest graph $G^{*}$ by solving~\eqref{eq:minimax},~\eqref{eq:minimaxLP1} or~\eqref{eq:minimaxLP2}. Instead of minimizing over $n_\lambda$ whilst keeping $\gamma,\beta$ fixed, in the next step we now fix $n_\lambda$ and try to maximize the objective over $\gamma,\beta$, i.e., we want to 
\begin{equation}
    \begin{aligned}
        &\max_{\gamma_d,\beta_d} &&  \frac{\sum_{\lambda} n_\lambda (G) f_{d,\lambda}(\beta_{d},\gamma_{d})}{\sum_{\lambda} n_\lambda (G) c_\lambda} \\
    &\text{s.t. } &&\gamma_d,\beta_d \in [0,2\pi),
    \end{aligned}
    \label{eq:max_h}
\end{equation}
for all $d\in\{1,2,3,4\}$.
For these new values of $\gamma^{*},\beta^{*}$, there might exist some other graph $G^{**}$ that is more difficult than the one we originally obtained. After we obtained $G^{**}$ we can again try to find new parameters $\beta^{**},\gamma^{**}$. This suggests the use of an iterative procedure for establishing the bound. We do not know whether this procedure converges, but it  doesn't necessarily have to: any combination of $\gamma,\beta$ has its own lower bound that holds specifically for these parameters, and can therefore be used as our result. Convergence would only suggest something about the tightness of this bound, guaranteed global convergence would mean that no better values of $\gamma,\beta$ exist. In our iterative procedure, the best parameter combinations we found are listed in Table~\ref{tab:parameters}:

\begin{table}[H]
\centering
\begin{tabular}{l|llll}
 & $d=1$ & $d = 2$ & $d=3$ & $d=4$ \\ \hline
$\gamma^*_d$ & $-$ &  $2.857$ & $2.773$ & $2.682$ \\
$\beta^*_d$ & $-$ & $0.4833$ & $0.1310$ & $0.1435$
\end{tabular}
\caption{Parameter values for different $d$ at which we were able to obtain the results of Theorem \ref{thm:alpha} and Conjecture \ref{con:alpha}. At $d=1$ the algorithm has only one state and hence no parameters.}
\label{tab:parameters}
\end{table}
For the parameter combinations in Table~\ref{tab:parameters}, \eqref{eq:minimaxLP1} can be solved to an arbitrary precision, which establishes the proof of the following theorem:
\begin{thm}
For 3-regular graphs $G=(V,E)$, where $G$ is not the complete $N=4$ graph, initial points $\gamma^*_d,\beta^*_d$ exist such that at $p=1$ the QAOA that loops over the clusters gives a $0.6367$-approximation algorithm, even without the classical optimization loop.
\label{thm:alpha}
\end{thm}
However, as stated before, this bound is too strict as the objective function of $G$ is overestimated. Using COBYLA as a local optimizer, we numerically observed that solving~\eqref{eq:minimax} never resulted in a bound lower than $0.6699$. This is equal to the upper bound we find by solving~\eqref{eq:minimaxLP2} (also $0.6699$). Since COBYLA does not guarantee a global minimum and~\eqref{eq:minimaxLP2} only considers a subset of all possible graphs, we can only conjecture that this is the actual lower bound:
\begin{conj}
For 3-regular graphs $G=(V,E)$, where $G$ is not the complete $N=4$ graph, initial points $\gamma^*_d,\beta^*_d$ exist such that at $p=1$ the QAOA that loops over the clusters gives a $0.6699$-approximation algorithm, even without the classical optimization loop. 
\label{con:alpha}
\end{conj}
For this particular combination of sub-graphs belong to the hardest graph $G^*$, the classical optimization step actually does not make much of a difference: our results show that solving~\eqref{eq:max_h} with  $G = G^* $ results in
\begin{equation*}
    \begin{aligned}
        &\max_{\gamma_d,\beta_d} &&\max_{d}  \frac{\sum_{\lambda} n_\lambda (\mathcal{G^*}) f_{d,\lambda}(\gamma_{d},\beta_{d})}{\sum_{\lambda} n_\lambda (\mathcal{G^*}) c_\lambda} \approx 0.674,
    \end{aligned}
\end{equation*}
which is only a small improvement on the conjectured bound. This shows the quality of these values of $\gamma^{*},\beta^{*}$ (as well as the hardness of the graph $G^*$). Additionally, this also provides further indication that QAOA, when having access to good initial points, can also be used without the classical optimization step~\cite{Brandao2018,Streif20192}.

\subsection{Performance bounds for $p>1$?}
Unfortunately, we do not observe evidence for the existence of a trivial graph hierarchy as  Wurtz and Love proved (and conjectured) for MAXCUT at $p\leq2$ ($p>2)$~\cite{Wurtz2020}. In fact, when we consider their large loop conjecture~\cite{Wurtz2020} for our problem, we find that for our used $\gamma,\beta$ we obtain a worst-case approximation ratio of $\alpha=0.693$, which is significantly larger than the bounds in Theorem~\ref{thm:alpha} and Conjecture~\ref{con:alpha}. Therefore, we do not conjecture the structure of the most difficult graph at any $p$, which would ease the determination of lower bounds at larger $p$. If we are to use the same method as we used for $p=1$, we will have to consider of the order of $123\cdot 2^{13} \approx 10^6$ different sub-graphs (not taking symmetries into account). We can reduce this number by exploiting symmetries, but since determining the energy of the largest sub-graph (14 nodes) is computationally very expensive we do not attempt to determine bounds for $p>1$. In fact, a back-of-the-envelope estimation to the amount of computing hours needed to execute such a computation -- again not utilizing symmetries -- shows that we would need on the order of 10 million computing hours.\\


\section{\label{app:RQC_BuildingBlocks} Details of the experimental building blocks}

\subsection{Dynamics of a driven three-level system}

\begin{figure}[t!]
	\centering
	\includegraphics[width=0.3\textwidth]{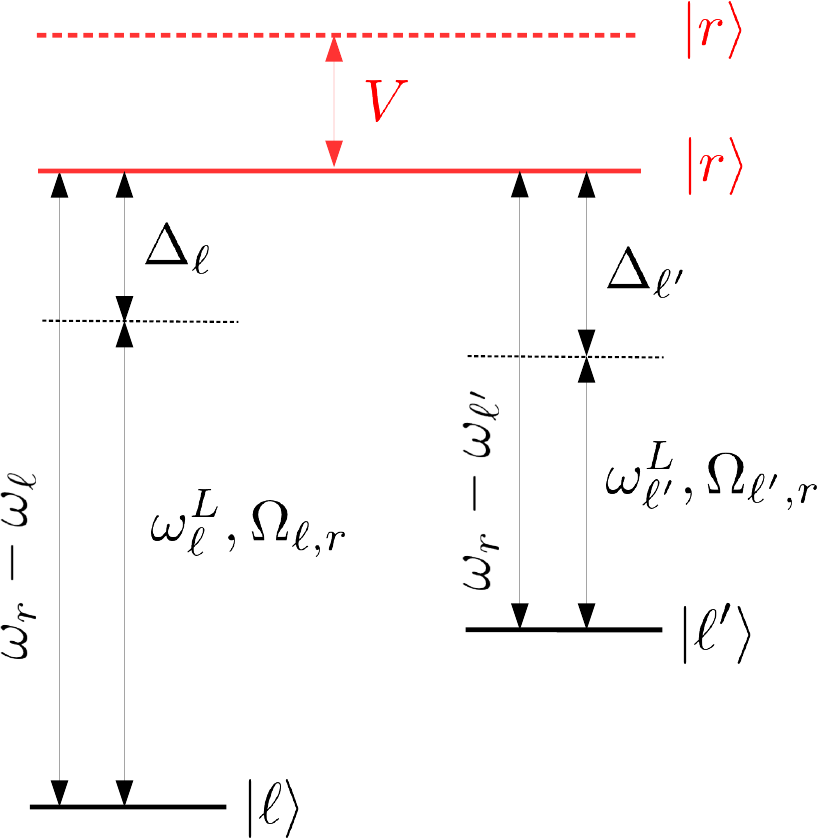}
	\caption{
		Scheme of a driven three-level system in a so-called $\Lambda$-configuration. Two qudit states $\ket{\ell}, \ket{\ell'}$ of energies $\omega_{\ell},\omega_{\ell'}$ are coupled with laser light of frequencies $\omega_{\ell}^L$,$\omega_{\ell'}^L$ and Rabi frequencies $\Omega_{\ell,r}, \Omega_{\ell',r}$ to a Rydberg state $\ket{r}$ of energy $\omega_r$. The Rydberg state can be additionally shifted by energy $V(R)$ if an atom at distance $R$ is excited to a Rydberg state (we recall we use $\hbar=1$).
	}	
	\label{fig:lambda_scheme}
\end{figure}

Here we briefly review the standard derivation of the unitaries Eqs.~(\ref{eq:U2level}),(\ref{eq:U3level}). Let us consider the level scheme depicted in Fig.~\ref{fig:lambda_scheme}. Two qudit states $\ket{\ell}, \ket{\ell'}$ of energies $\omega_{\ell}$,$\omega_{\ell'}$ are coupled with lasers of frequencies $\omega_\ell^L$,$\omega_{\ell'}^L$ and Rabi frequencies $\Omega_{\ell,r}$,$\Omega_{\ell',r}$ to a Rydberg state $\ket{r}$ of energy $\omega_r$. Furthermore, we allow for the Rydberg state to be shifted by energy $V$ if a nearby atom is in a Rydberg state ($V=V(R) = C_6/R^6$ for the atoms separated by a distance $R$ and $C_6$ is the Van der Waals coefficient, cf. Sec.~\ref{sec:RQC_BuildingBlocks}). Within the rotating wave approximation and in the frame where the levels $\ket{\ell},\ket{\ell'}$ rotate at the laser frequencies $\omega_{\ell}, \omega_{\ell'}$ respectively~\cite{scully1999quantum, steck}, the corresponding Hamiltonian reads
\beq
    H = \sum_{j=\ell,\ell'} \Omega_{j,r} \ket{j}\bra{r} + {\rm H.c.} + \Delta_j \ket{j}\bra{j} + V \ket{r}\bra{r}.
    \label{eq:H_3level}
\eeq
Here $\Delta_j = \omega_j^L - (\omega_r - \omega_j)$ is the single-photon detuning between the qudit state $j=\ell,\ell'$ and the Rydberg state $\ket{r}$ (note our sign convention, where $\Delta_j < 0$ refers to a red-detuned laser beam). Written in the $\{ \ket{\ell}, \ket{\ell'},\ket{r} \}$ basis, the Hamiltonian (\ref{eq:H_3level}) is
\beq
	H = 
	\begin{pmatrix}
		\Delta_\ell & 0 & \Omega_{\ell,r} \\
		0 & \Delta_{\ell'} & \Omega_{\ell',r} \\
		\Omega_{\ell,r}^* & \Omega_{\ell',r}^* & V
	\end{pmatrix}.
	\label{eq:H_3level_mat}
\eeq
The associated unitary operator $U = {\rm e}^{-i t H}$ can in principle be obtained analytically by diagonalizing (\ref{eq:H_3level_mat}), leading to a cubic equation for the eigenvalues. The situation simplifies at two-photon resonance $\Delta_{\ell} = \Delta_{\ell'} = 0$. 
In the limit of large interaction, $V \rightarrow \infty$, we get
\beq
	U = 
	\begin{pmatrix}
		1 & 0 & 0 \\
		0 & 1 & 0 \\
		0 & 0 & {\rm e}^{-i V t}
	\end{pmatrix}.
\eeq
Denoting $\Omega_0 = \Omega_{\ell,r}, \Omega_1 = \Omega_{\ell',r}$ for the ease of notation, the $\ket{r}\bra{r}$ element of $U$ for a general $V$ reads
\beq
	U_{rr} = e^{-\frac{1}{2} i t V} \left(\cos \left(\frac{1}{2} t \sqrt{V^2+4 \Omega ^2}\right)-\frac{i V \sin \left(\frac{1}{2} t \sqrt{V^2+4 \Omega ^2}\right)}{\sqrt{V^2+4 \Omega
   ^2}}\right),
\eeq
where
\beq
	\Omega = \sqrt{|\Omega_0|^2 + |\Omega_1|^2}.
\eeq
Requiring that the Rydberg population returns to zero at the end of the evolution, if it was initially unpopulated, is equivalent to $|U_{rr}|=1$. This leads to the condition
\beq
	t = \frac{2 m \pi}{\sqrt{V^2+4 \Omega ^2}}, \; m \in {\mathbb N}.
\eeq
Specifically, in the absence of interaction ($V=0$) and for $m$ odd we get the expression Eq.~(\ref{eq:U3level})
\beq
	U^{\rm 3-level}_{\ell,\ell'} =	  -\begin{pmatrix}
     \cos \frac{\theta}{2} & {\rm e}^{i \varphi} \sin \frac{\theta}{2} & 0\\
     {\rm e}^{-i \varphi} \sin \frac{\theta}{2} & -\cos \frac{\theta}{2} & 0 \\
     0 & 0 & 1
    \end{pmatrix}, \nonumber
\eeq
with $\theta$ and $\varphi$ defined in Eqs.~(\ref{eq:theta_def}).

The unitary Eq.~(\ref{eq:U2level}) is obtained in an analogous way when considering a two-level system with levels $\ket{\ell},\ket{\ell'}$ coupled by a laser with Rabi frequency $\Omega_{\ell,\ell'}$ and described by the Hamiltonian
\beq
    H = \Omega_{\ell,\ell'} \ket{\ell}\bra{\ell'} + {\rm H.c.} + \Delta \ket{\ell'}\bra{\ell'},
\eeq
where $U={\rm e}^{-i t H}$ is again evaluated on resonance $\Delta=0$.

\subsection{Derivation of the qudit controlled-phase gates}

In this section we describe a systematic construction of the qudit controlled-phase gates Eqs.~(\ref{eq:CP_symmetric}) and (\ref{eq:CP_Kj}). 

\subsubsection{On qudit controlled-phase gate Eq.~(\ref{eq:CP_symmetric})}
\label{app:RQC_gate_symm}

Let us start with the gate (\ref{eq:CP_symmetric}) acting on a qubit, $d=2$. First we note that the action of the unitary ${\cal U}^{(2|1)}_0$ corresponds to
\beq
    {\cal U}^{(2|1)}_0 \ket{\ell}\ket{\ell'} = 
    \begin{cases}
        {\rm e}^{i\frac{\gamma}{2}} \ket{\ell} \ket{\ell'} \;\; {\rm if} \; (\ell,\ell') = (1,0) \\
        \ket{\ell} \ket{\ell'} \;\; {\rm otherwise}.
    \end{cases}
\eeq
It then follows that for a qubit, the controlled-phase gate expressed in the basis $\{ \ket{00},\ket{01},\ket{10},\ket{11} \}$ reads 
\beq
{\rm CP}(\frac{\gamma}{2}) = {\cal U}^{(2|1)}_0 {\cal U}^{(2|1)}_1 = {\rm diag}(1,{\rm e}^{i \gamma/2},{\rm e}^{i \gamma/2},1)
\label{eq:CP_qubit_simple}
\eeq
(in this case, an alternative sequence ${\cal U}^{(1|2)}_0 {\cal U}^{(2|1)}_0$ [or ${\cal U}^{(2|1)}_0 {\cal U}^{(1|2)}_0$] yields the same result). Furthermore, the cost of the controlled-phase gate (\ref{eq:CP_qubit_simple}) is $[{\rm CP}]=2=d$. 

Motivated by the construction in Eq.~(\ref{eq:CP_qubit_simple}), we next wish to extend it to qutrits. Here the situation becomes slightly more involved. Observing that the gate ${\cal U}^{(q_t|q_c)}_\ell$ is always diagonal, i.e. introducing at most a phase for each element $\ket{\ell} \ket{\ell'}$ of the two-qutrit system, we introduce the following graphical notation for a representation of a general \emph{diagonal} gate acting in the two-qutrit space
\beq
	{\cal U} \rightarrow
	\parbox{0.3\textwidth}{
	\diag{22}{02}{20}{21}{12}{10}{00}{11}{01} \raisebox{0.9cm}{,}
	\label{eq:Uqutrit_notation}
	}
\eeq
where the labels on the right-hand side denote the coordinates of the diagonal element $\bra{\ell \ell'} {\cal U} \ket{\ell \ell'}$ of ${\cal U}$. Since ${\cal U}$ is diagonal with elements, which are either 1 or a pure phase ${\rm e}^{i \varphi}$, we shall use the notation $\bullet$ for $1$ and $\varphi$ for the pure phase for easier readability.
With the help of (\ref{eq:Uqutrit_notation}), we can list the action of the possible six unitaries ${\cal U}_\ell^{(q_t|q_c)}$ acting on the two qutrits:
\beqa
    {\cal U}^{(1|2)}_2 &=&
    \parbox{0.3\textwidth}{
    \diag{\bullet}{\bullet}{\gamma/2}{\gamma/2}{\bullet}{\bullet}{\bullet}{\bullet}{\bullet}
    }
    \;\; {\cal U}^{(2|1)}_2 = 
    \parbox{0.3\textwidth}{
    \diag{\bullet}{\gamma/2}{\bullet}{\bullet}{\gamma/2}{\bullet}{\bullet}{\bullet}{\bullet}
    }
    \nonumber \\
    {\cal U}^{(1|2)}_1 &=&
    \parbox{0.3\textwidth}{
    \diag{\bullet}{\bullet}{\bullet}{\bullet}{\gamma/2}{\gamma/2}{\bullet}{\bullet}{\bullet}
    }
    \;\;
    {\cal U}^{(2|1)}_1 =
    \parbox{0.3\textwidth}{
    \diag{\bullet}{\bullet}{\bullet}{\gamma/2}{\bullet}{\bullet}{\bullet}{\bullet}{\gamma/2}
    }
    \nonumber \\
    {\cal U}^{(1|2)}_0 &=& 
    \parbox{0.3\textwidth}{
\diag{\bullet}{\gamma/2}{\bullet}{\bullet}{\bullet}{\bullet}{\bullet}{\bullet}{\gamma/2}
}
\;\;
{\cal U}^{(2|1)}_0 =
\parbox{0.3\textwidth}{
\diag{\bullet}{\bullet}{\gamma/2}{\bullet}{\bullet}{\gamma/2}{\bullet}{\bullet}{\bullet} \raisebox{0.9cm}{.}
}
\label{eq:Uqutrit}
\eeqa
It then follows that the controlled-phase gate can be obtained by concatenation of all the gates in (\ref{eq:Uqutrit})
\beq
\centering
	{\rm CP}(\gamma) = {\cal U}^{(1|2)}_2 {\cal U}^{(2|1)}_2 {\cal U}^{(1|2)}_1 {\cal U}^{(2|1)}_1 {\cal U}^{(1|2)}_0 {\cal U}^{(2|1)}_0 =
	\parbox{0.3\textwidth}{
\diag{\bullet}{\gamma}{\gamma}{\gamma}{\gamma}{\gamma}{\bullet}{\bullet}{\gamma} \raisebox{0.9cm}{,}
}
\label{eq:CP_qutrit}
\eeq
which corresponds to Eq.~(\ref{eq:CP_gate}) up to a global phase. It is then a straightforward exercise to verify that the manifestly symmetric form suggested by (\ref{eq:CP_qutrit}) generalizes to higher $d$ yielding the expression (\ref{eq:CP_symmetric}). In fact, it also applies to a qubit, yielding the cost $[{\rm CP}]=4=2d$, which is twice as large as that of (\ref{eq:CP_qubit_simple}). For the ease of exposition and compactness of notation, and motivated by the fact that we are mainly interested in applications with $d>2$, we have kept the controlled-phase gate (\ref{eq:CP_symmetric}) also for qubits in the main text.

\subsubsection{On qudit controlled-phase gate Eq.~(\ref{eq:CP_Kj})}

To understand the logic behind the construction of (\ref{eq:CP_Kj}), it is instructive to work out a specific example and we shall consider the simplest case beyond qubit, the qutrit. To appreciate the role of the qudit CX gates in (\ref{eq:CP_Kj}), we focus specifically on the equal-state elements $\ket{\ell}\ket{\ell}$ of the two-qutrit system. The evolution of these elements under the action of the controlled-phase gate (\ref{eq:CP_Kj}) can be symbolically written as
\beqa
    \ket{22} \xtofrom[{\rm CX}^{(1|2)}_{1,2|\neg 1}]{{\rm CX}^{(1|2)}_{1,2|\neg 1}} \ket{12} \xtofrom[{\rm CX}^{(1|2)}_{0,1|\neg 0}]{{\rm CX}^{(1|2)}_{0,1|\neg 0}} {\color{red} {\rm e}^{i \gamma}}\ket{02} {\color{red} \hookleftarrow P^{(1)}_0(\gamma)} \nonumber\\
    \ket{11} \xtofrom[\phantom{{\rm CX}^{(1|2)}_{1,2|\neg 1}}]{\phantom{  {\rm CX}^{(1|2)}_{1,2|\neg 1}}  } \ket{11} \xtofrom[\phantom{{\rm CX}^{(1|2)}_{0,1|\neg 0}}]{  \phantom{{\rm CX}^{(1|2)}_{0,1|\neg 0}}  } {\color{red} {\rm e}^{i \gamma}}\ket{01} {\color{red} \hookleftarrow P^{(1)}_0(\gamma)} \nonumber \\
    \ket{00} \xtofrom[\phantom{{\rm CX}^{(1|2)}_{1,2|\neg 1}}]{\phantom{  {\rm CX}^{(1|2)}_{1,2|\neg 1}}  } \ket{00} \xtofrom[\phantom{{\rm CX}^{(1|2)}_{0,1|\neg 0}}]{  \phantom{{\rm CX}^{(1|2)}_{0,1|\neg 0}}  } {\color{red} {\rm e}^{i \gamma}}\ket{00} {\color{red} \hookleftarrow P^{(1)}_0(\gamma)}.
    \label{eq:CP_Kj_scheme}
\eeqa
Here, the CX gates (listed only in the first line to avoid cluttering) act in the direction indicated by the arrows. We thus see that the purpose of the CX gates is to bring the equal state elements to the elements of the form $\ket{0 \ell},\; \forall \ell$, i.e. ${\rm CX}^{(1|2)}_{0,1|\neg 0} {\rm CX}^{(1|2)}_{1,2|\neg 1} \ket{\ell \ell} \rightarrow \ket{0\ell}, \; \ell=0,1,2$. Since there are at most three states of the form $\ket{0 \ell}$, the states $\ket{02},\ket{01},\ket{00}$ on the right of (\ref{eq:CP_Kj_scheme}) exhaust all such states. The phase gate $P^{(1)}_0$ then imprints a phase $\gamma$ to all these states, which is highlighted in red, and acting with the CX gates backwards yields the controlled-phase gate (\ref{eq:CP_Kj}). Similarly to the construction in Sec.~\ref{app:RQC_gate_symm} above, it is straightforward to verify that the scheme (\ref{eq:CP_Kj_scheme}) holds also for higher $d$.



\section{Simulating complete graphs using only nearest-neighbour interactions}
\label{app:HW}

For first experiments, it might be prudent to focus on problem instances where the problem graph matches the topology of the quantum computer well. Our simulations considered complete graphs, so for those it is natural to ask: How many swap operations are required to let a limited-interaction quantum computer execute the algorithm on a complete graph? Or, if the hardware allows swaps on distinct nearest-neighbour pairs of qudits to be performed in parallel, how many \emph{layers} of swaps?

We can construct a simple sequence of swaps that achieves this task using not too many layers:
\begin{prop}
There exists a sequence of swaps such that the complete graph can be simulated on a line graph in $n-2$ layers of swaps. The total swap count for this protocol is $\frac{(n-1)(n-2)}{2}$.
\end{prop}
\begin{proof}
First assume $n$ is even (the argument follows the same structure if $n$ is odd).
The sequence will be generated by alternating the two following sets of swaps: $\pi=\{(1,2),(3,4),\dots,(n-1,n)\}$, $\sigma=\{(2,3),(4,5),\dots,(n-2,n-1)\}$. First observe what happens when applying $\pi$ and $\sigma$ to a qudit starting in an odd position $i$. If $i\neq n-1$, this qudit will move to $i+1$ because of application of $\pi$, and then to $i+2$ by $\sigma$. If $i = n-1$, the qudit will move to position $n$. Similarly, for a qudit starting in even position $j$ will move to position $j-2$, except that the qudit at position 2 will move to position 1.

It is easily checked that after $n-2$ layers of swaps, all pairs of qudits have been nearest-neighbours at some stage of the process.
\end{proof}

For a more thorough analysis of swapping sequences for executing all-to-all interactions, see also Ref.~\cite{Babbush_2018_PRX} and Ref.~\cite{Kivlichan_2018_PRL}, where variants of the previous construction are analysed in-depth.

We can easily see that the previous strategy is close to optimal, and that no strategy can exist that is much better: by a simple counting argument the number of layers of swaps required to enable all interactions of a complete graph on a line could not be much lower. This motivates us to look for the corresponding lower bound.

\begin{prop}
Consider a quantum computer for which the graph of possible interactions is a line of $n$ qudits. Then at least $\frac{1}{2}\binom{n}{2}-o(n)$ SWAP gates are required to enable all-to-all interactions. Additionally, at least $\frac{n}{2}-1$ layers of SWAP gates are required.
\end{prop}
\begin{proof}
We can lower bound the number of necessary swaps by a simple counting argument: The complete graph has $\binom{n}{2}$ edges, so that number of interactions will be required. A line of $n$ qubits has $n-1$ edges, which is how many interactions are possible before the first SWAP gate. Next, every SWAP gate gives two qubits a new neighbour, enabling at most two new interactions. The required swap count therefore is at least $\frac{1}{2}\left[ \binom{n}{2}-(n-1) \right]$.

For the number of layers, note again between a layer of SWAP gates at most $n-1$ nearest-neighbour interactions are possible. This implies a lower bound of $\frac{\binom{n}{2}}{n-1} = \frac{n}{2}$ of layers, hence $\frac{n}{2}-1$ layers of SWAP gates to be able to have all interactions.
\end{proof}


\section{On experimental errors}
\label{app:Errors}
\begin{figure}[t!]
	\centering
	\includegraphics[width=0.45\textwidth]{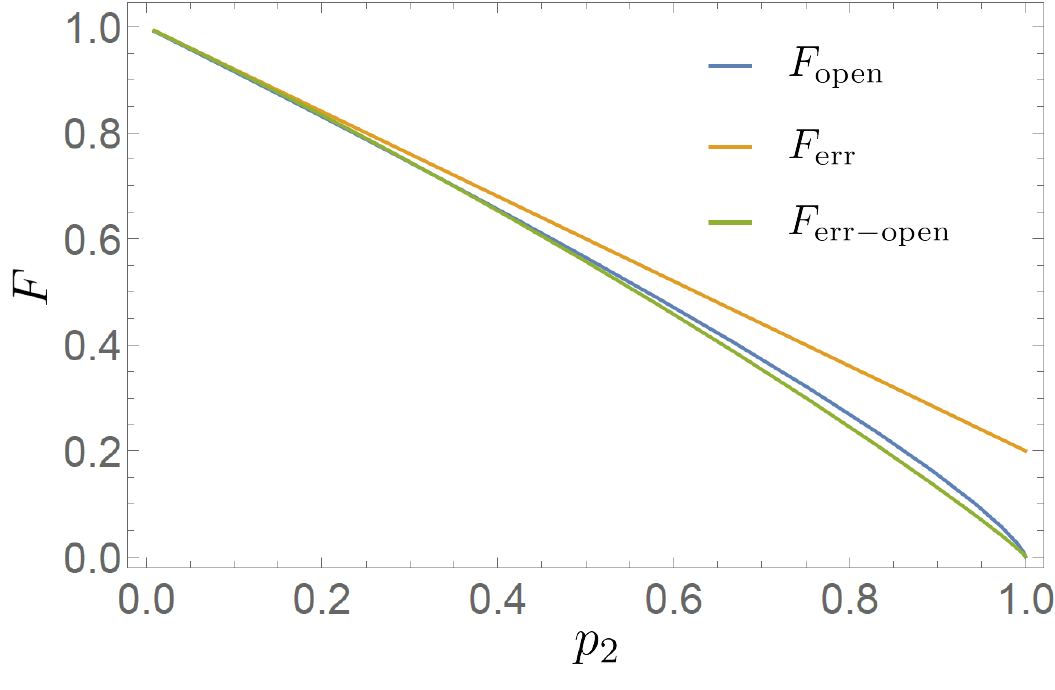}
	\caption{
		Fidelities (\ref{eq:F_err}), (\ref{eq:F_open}) and (\ref{eq:F_err_open}) [solid orange, blue and green lines] as a function of the 2-gate error $p_2$, cf. Eqs.~(\ref{eq:err_model}) and (\ref{eq:p_vs_decay}), for $\ket{\psi_{\rm in}} = \ket{+^2}^{\otimes 2}$ and CP gate phase $\gamma=\pi/2$. The approximate agreement between $F_{\rm open}$ and $F_{\rm err}$ for small values of $p_2$ was obtained by setting $\eta \approx 1.5$ in (\ref{eq:p_vs_decay}).
	}	
	\label{fig:Gerr_Qubit}
\end{figure}
The gate errors for Rydberg atoms have been analysed extensively, e.g.\ in~\cite{Saffman_2005_PRA}, and verified experimentally, e.g. in~\cite{Madjarov_2020_NatPhys}. 
The total error budget consists of various contributions including the collisions of the atoms with background gas, phase and intensity laser noise, light scattering from both the trapping and resonant lasers operating the processor, spurious level shifts from fluctuating electric and magnetic fields or atom motion in the tweezer traps. A detailed analysis of the individual contributions of these error sources goes beyond the scope of the present work. Given that the total 2-gate error dominates over the 1-gate one~\cite{Saffman_2005_PRA} and is one of the bottlenecks for scalability across different platforms, we consider only the 2-gate error. Specifically, we consider the error due to the spontaneous decay from the Rydberg levels. 

Let us start with a brief overview of the situation and let us consider ${\mathbb P}$ as the qudit manifold (we recall we consider the excitation to a Rydberg state as a resonant two step process through the long-lived ${\mathbb P}$ manifold, cf. Sec. \ref{sec:RQC_BuildingBlocks}). For the corresponding Rydberg manifold $\ket{{\mathbb Ry}} = \ket{{\rm n} {}^3{\rm S}_1}$, the Rydberg atom decays approximately with branching ratios $\approx$1:4 to the $(\rm n p){}^3{\rm P_J}$ and $(\rm 5 p){}^3{\rm P_J}$ manifolds respectively, where $n>5$. Out of the decay to $(\rm 5 p){}^3{\rm P_J}$, it branches further with ratios $\approx$1:3:5 to the $J=0,1,2$ manifolds and then further among the $F-$manifolds and the respective $m_F$ states (allowed by the selection rules). Consequently, the probability of the Rydberg state to decay back to the qudit manifold (here $\ket{{}^3{\rm P}_2}$) -- but not necessarily to the same qudit state -- is $\approx 0.4$~\cite{Sibalic_CompPhysComm_2017,ARC_package}. 

The complete dynamics can be described by the standard means of a master equation for amplitude decay including all the decay channels. However, since the probability to decay outside the qudit manifold $ \gtrsim 0.6$, we consider a simplified dynamics of a pair of interacting qudits $q=1,2$ given by 
\beq
	\dot{\rho} = i[\rho,H]
	+ \frac{\Gamma}{2} \sum_{q=1,2}\sum_{\ell=0}^{d-1} 2 {\sigma^{(q)}}^-_\ell \rho {\sigma^{(q)}}^+_{\ell} - \{ {\sigma^{(q)}}^+_\ell {\sigma^{(q)}}^-_\ell, \rho \}
	\label{eq:rho_master}
\eeq
where ${\sigma^{(q)}}^+_\ell = \ket{r^{(q)}_\ell}\bra{\rm aux}$, $\{ \cdot, \cdot \}$ is the anticommutator and a single auxiliary level $\ket{\rm aux}$ is used to model the decay at rate $\Gamma$ outside the qudit manifold.

We seek to compare the effect of the realistic errors described by the master equation (\ref{eq:rho_master}) to the unitary error model introduced above. Let us denote
\beq
	\rho_{\rm id} \equiv {\rm CP} \ket{\psi_{\rm in}} \bra{\psi_{\rm in}} {\rm CP}^\dag
\eeq
the ideal state obtained upon an action of the qudit controlled-phase gate (\ref{eq:CP_gate}) on some pure state $\ket{\psi_{\rm in}}$. We also define the fidelity between two quantum states with density matrices $\rho,\sigma$ as usual,
\beq
	F(\rho,\sigma) = {\rm Tr} \left[ \sqrt{\sqrt{\rho} \sigma \sqrt{\rho}} \right].
\eeq
Next, we define the average fidelity of a state corresponding to the error model (\ref{eq:err_model}) as
\beq
    F_{\rm err}(p_2) = (1-p_2) + \frac{p_2}{|{\mathbb U}^{\otimes 2}\setminus \{ \mathds{1} \}|} \sum_{U \in {\mathbb U}^{\otimes 2}\setminus \{ \mathds{1} \} } F\left(U \rho_{\rm id} U^\dag, \rho_{\rm id} \right).
    \label{eq:F_err}
\eeq
Similarly, denoting the density matrix resulting from the open dynamics given by the master Eq.~(\ref{eq:rho_master}) as $\rho_{\rm open}$, we define the fidelity
\beq
    F_{\rm open} = F(\rho_{\rm open},\rho_{\rm id})
    \label{eq:F_open}
\eeq
and the average fidelity between $\rho_{\rm open}$ and the states generated by (\ref{eq:err_model})
\beq
    F_{\rm err-open} = (1-p_2) F\left(\rho_{\rm id}, \rho_{\rm open} \right)
    + \frac{p_2}{|{\mathbb U}^{\otimes 2}\setminus \{ \mathds{1} \}|} \sum_{U \in {\mathbb U}^{\otimes 2}\setminus \{ \mathds{1} \} } F\left(U \rho_{\rm id} U^\dag, \rho_{\rm open} \right).
    \label{eq:F_err_open}
\eeq
The resulting fidelities depend on $\ket{\psi_{\rm in}}$ and they in principle vary in the course of the algorithm. Furthermore, there is no direct unambiguous identification between the the error probability $p_2$ of the error model (\ref{eq:err_model}) and that of the open dynamics (\ref{eq:rho_master}). Denoting the gate time of the cost unitary (\ref{eq:CP_Kj}) as $t_{\rm CP}$, the decay probability of an excited atom $\propto 1-{\rm e}^{-\Gamma t_{\rm CP}}$. This motivates a parametrization
\beq
    p_2 = 1-{\rm e}^{-\eta \Gamma t_{\rm CP}},
    \label{eq:p_vs_decay}
\eeq
where we have introduced a phenomenological factor $\eta$.

For illustration, in Fig.~\ref{fig:Gerr_Qubit} we show the fidelities (\ref{eq:F_err}), (\ref{eq:F_open}) and (\ref{eq:F_err_open}) as a function of the error probability $p_2$, cf. Eq.~(\ref{eq:err_model}), for the ``isotropic'' initial state $\ket{\psi_{\rm in}} = \ket{+^2}^{\otimes 2}$, $\ket{+^d} = 1/\sqrt{d} \sum_{\ell=0}^{d-1} \ket{\ell}$.

Denoting ${\cal P}_{\bb Q}$ the projector on the qudit subspace, ${\mathbb Q} = {\mathbb P}$, we have also verified that the state evolution subject to the open dynamics (\ref{eq:rho_master}) satisfies to a good accuracy ${\cal P}_{\bb Q} \rho_{\rm open} {\cal P}_{\bb Q} \approx {\rm e}^{-\Gamma t_{\rm CP}} \rho_{\rm id}$. This corresponds well to the naive guess that the resulting density matrix is just the rescaled ideal matrix $\rho_{\rm id}$ due to the decay with rate $\Gamma$ outside the qudit subspace. 

In summary, while the error model (\ref{eq:err_model}) seems to capture qualitatively the decrease of the state fidelity, it clearly cannot account for the dynamics outside the qudit subspace and more rigorous analysis of the errors in the context of the QAOA is desirable, see also \cite{Henriet_2020_PRA} for related developments.


\end{document}